\documentclass[runningheads]{llncs}

\usepackage[T1]{fontenc}
\usepackage[colorlinks, citecolor=blue, linkcolor=blue]{hyperref}
\usepackage{color}

\usepackage{amsmath,amssymb}
\usepackage{mathtools}
\usepackage{cleveref}
\usepackage{booktabs}
\usepackage{tikz}
\usetikzlibrary{automata, petri, positioning, fit, arrows, arrows.meta}
\usepackage{pgfplots}
\pgfplotsset{compat=1.18}
\usepackage[ruled, noend]{algorithm2e}
\Crefname{algocf}{Algorithm}{Algorithms}
\usepackage{thm-restate}
\usepackage{bm}
\usepackage{colortbl}

\usepackage{macros}
\usepackage{plot_macros}

\renewcommand{\orcidID}[1]{\href{https://orcid.org/#1}{\,\raisebox{-1pt}{\includegraphics[width=8pt]{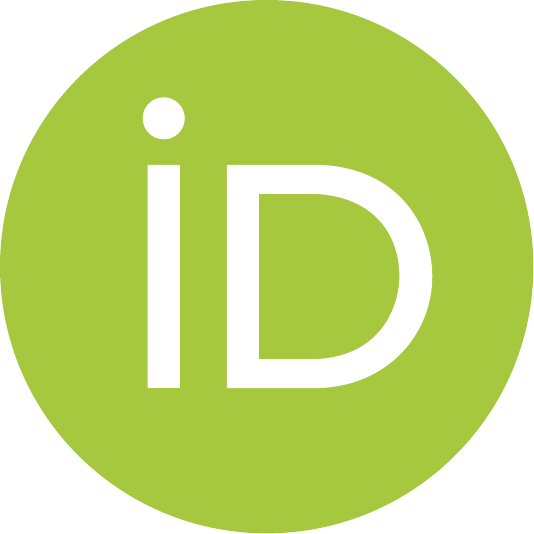}}}}

\begin{document}

\title{Weakly acyclic diagrams: A data structure for infinite-state symbolic verification\thanks{M.~Blondin was supported by a Discovery Grant from the Natural Sciences and Engineering Research
Council of Canada (NSERC).}}
\titlerunning{WADs: a data structure for infinite-state symbolic verification}

\author{Michael Blondin\inst{1}\orcidID{0000-0003-2914-2734} \and
  Michaël Cadilhac\inst{2}\orcidID{0000-0001-9828-9129} \and Xin-Yi
  Cui\inst{3} \and Philipp
  Czerner\inst{3}\orcidID{0000-0002-1786-9592} \and Javier
  Esparza\inst{3}\orcidID{0000-0001-9862-4919} \and Jakob
  Schulz\inst{3}\orcidID{0009-0007-2596-0920}}

\authorrunning{M. Blondin et al.}

\institute{%
  Universit\'{e} de Sherbrooke, Canada
  %\email{michael.blondin@usherbrooke.ca}
  \and
  DePaul University, Chicago, USA
  %  \email{michael@cadilhac.name}
  \and
  Technical University of Munich, Germany
%  \email{\{czerner,esparza\}@in.tum.de}
%  \email{jakob.schulz@tum.de}
}

%% Title
\maketitle

%% Abstract
\begin{abstract}
  Ordered binary decision diagrams (OBDDs) are a fundamental data
  structure for the manipulation of Boolean functions, with strong applications to finite-state symbolic model checking. OBDDs allow for efficient algorithms using top-down dynamic programming. From an automata-theoretic perspective, OBDDs essentially are minimal deterministic finite automata recognizing languages whose words have a fixed length (the arity of the Boolean function). We introduce weakly acyclic diagrams (WADs), a generalization of OBDDs that maintains their algorithmic advantages, but can also represent 
infinite languages. We develop the theory of  WADs and show that they can be used for symbolic model checking of various models of infinite-state systems. 
\end{abstract}

%% Contents
\section{Introduction}
\label{sec:intro}
Ordered binary decision diagrams (OBDDs) are a fundamental data structure for the manipulation of Boolean functions~\cite{Akers78,Bryant86}, widely used to support symbolic model checking of finite-state systems~\cite{Bryant18,ChakiG18}. In this approach to model checking,  sets of configurations of a system are encoded as Boolean functions of fixed arity, say $n$, and the transition relation is encoded as a Boolean function of arity $2n$. After fixing a total order on the  Boolean variables, the OBDD representation of a Boolean function is unique up to isomorphism, and one can implement different model checking algorithms on top of an OBDD library.

The fundamental operations on sets and relations used by symbolic model-checking algorithms are---besides union, intersection, and complement of sets and relations---the \textit{pre} and \textit{post} operations that return the immediate predecessors or successors of a given set of configurations with respect to the transition relation. All of them are efficiently implemented on OBDDs using top-down dynamic programming: a call to an operation for functions of arity $n$ invokes one or more calls to the same or other operations for functions of arity $n-1$; intermediate results are stored using memoization~\cite{And98,Bryant18,ChakiG18}.

OBDDs can be presented in automata-theoretic terms. A Boolean function $f(x_1, \ldots, x_n)$ with order $x_1 < \cdots < x_n$ is representable by the language $L_f \defeq \{b_1 \cdots b_n \in \{0, 1\}^n : f(b_1, \ldots, b_n) = \true\}$. The OBDD for $f$ with this variable order is very close to the minimal DFA (MDFA) for $L_f$.\footnote{See \eg~\cite[Chap.~6]{EB23}, which even introduces a slight generalization of DFAs such that the OBDD \textit{is} the unique minimal deterministic generalized automaton for $L_f$.} From this point of view, OBDDs are just MDFAs recognizing \emph{fixed-length languages}---languages whose words have all the same length---and operations on OBDDs  are just operations on fixed-length languages canonically represented by their unique MDFAs.  

The automata-theoretic view of OBDDs immediately suggests a generalization from fixed-length languages to arbitrary regular languages: canonically represent a regular language by its unique MDFA, and manipulate MDFAs using automata-theoretic operations. (This is for example the approach of the tool MONA~\cite{monamanual2001}.) However, in this generalization crucial advantages of OBDD algorithms get lost.  Consider the implementation of the \textit{post} operation that, given a MDFA $\A$ recognizing a set of configurations and the MDFA $\T$ recognizing the transition relation, computes the MDFA for the set of immediate successors of $L(\A)$ with respect to $L(\T)$ (see \eg~\cite[Chap.~5]{EB23}). The algorithm proceeds in three steps. First, it constructs an NFA $\A'$ recognizing the set of immediate successors; then, it applies the powerset construction to determinize $\A'$ into a DFA $\A''$;  finally, it minimizes $\A''$ to yield the final MDFA $\A'''$. The crucial point is that the intermediate DFA $\A''$ can be \emph{exponentially larger} than $\A'''$, and so, in the worst case, \textit{post} uses exponential space in the size of $\A'''$, the final output. This is not the case for the OBDD implementation for fixed-length languages, which constructs $\A'''$ \emph{directly} from $\A$ and $\T$, using only space $\O(|\A'''|)$.

This observation raises the question of whether OBDDs can be generalized beyond fixed-length languages \emph{while maintaining the advantages of OBDD algorithms}. We answer it  in the affirmative. We introduce \emph{weakly acyclic diagrams (WADs)}, a generalization of OBDDs for the representation of \emph{weakly acyclic languages}. A language is weakly acyclic if it is recognized by a weakly acyclic DFA, and a DFA is weakly acyclic if every simple cycle of its state-transition diagram is a self-loop. So, for example, the language $a^* b a^*$ is weakly acyclic. Note that, contrary to fixed-length languages, weakly acyclic languages can be infinite.

We develop the theory of WADs and in particular show that small modifications to the top-down dynamic programming algorithms for OBDDs generalize them to WADs. We also present  a first application of WADs to symbolic infinite-state model checking. A fundamental technique in this area is the backwards reachability algorithm for well-structured transition systems~\cite{FinkelS01,AbdullaCJT96}. The algorithm computes the set of all predecessors of a given upward-closed set of configurations (with respect to the partial order making the system well-structured). For that, it iteratively computes the immediate predecessors of the current set of configurations, until a fixpoint is reached. Well-structuredness ensures that all intermediate sets remain upward-closed, and termination. We show that for many well-structured transition systems, including lossy channel systems~\cite{AbdullaJ93,AbdullaK95}, Petri nets~\cite{Murata89}, and broadcast protocols~\cite{EsparzaFM99}, upward-closed sets of configurations can be easily encoded as weakly acyclic languages, and so the backwards reachability algorithm can be implemented with WADs as data structure. 

We implemented our algorithms in \wadl{}, a prototype library for weakly acyclic diagrams. We conduct experiments on over 200 inputs for the backwards reachability algorithm, including instances of lossy channel systems, Petri nets, and broadcast protocols. We compare with some established tools using dedicated data structures and show that our generic approach is competitive. 

\paragraph{Related work.}

Weakly acyclic automata have been studied extensively. However, past work has focused on algebraic, logical, and computational complexity questions, not on their algorithmics as a data structure.

In~\cite{BF80}, Brzozowski and Fich showed that weakly acyclic
automata, called \emph{partially ordered automata} there, capture the
so-called $\mathcal{R}$-trivial languages. They credit Eilenberg 
for earlier work on $\mathcal{R}$-trivial monoids~\cite{Eil76}.
Weakly acyclic automata have also been called
\emph{extensive automata}, \eg\ in~\cite{Pin86}.

Weak acyclicity has also been defined in terms of \emph{nondeterministic} automata. 
As we shall see, one such definition yields a model as expressive as the deterministic one, while another one 
is more expressive. These models have sometimes been called \emph{restricted
partially ordered NFAs (rpoNFAs)} and \emph{partially ordered NFAs (poNFAs)}.
Schwentick, Thérien and Vollmer have shown that poNFAs
characterize level $3/2$ of the Straubing-Thérien hierarchy~\cite{STV01}. Bouajjani
\etal\ characterized the class of languages accepted by poNFAs as languages closed under permutation rewriting, \ie, languages described by so-called \emph{alphabetic pattern constraints}~\cite{BMT07}. As mentioned, \eg, 
in the introduction of~\cite{MK21}, further equivalences are known in
first-order logic.

More recently, Krötzsch, Masopust and Thomazo have studied
the computational complexity of basic questions for several types of partially ordered automata~\cite{KMT17,MK21}, and
Ryzhikov and Wolf have studied the complexity of problems for upper triangular Boolean matrices~\cite{RW24};
these matrices arise naturally from the transformation monoid of partially ordered automata. 

In the context of symbolic verification, there is work on
\emph{dedicated} data structures for the representation of
configuration sets for \emph{specific} systems. In particular,
Delzanno \etal\ introduced \emph{covering sharing trees} to represent
upward-closed sets of Petri net markings~\cite{DRB04}, which were
later extended to \emph{interval sharing trees} by
Ganty \etal~\cite{GMDKRV07}. Boigelot and Godefroid studied
\emph{queue-content decision diagrams} to analyze
channel systems~\cite{BG99}.

\paragraph{Structure of the paper.}

In \Cref{sec:wa}, we introduce and study weakly acyclic languages. In
\Cref{sec:ds,sec:oper}, we introduce weakly acyclic diagrams and
explain how to implement many operations. In \Cref{sec:model:check},
we explain how our framework can be used for infinite-state symbolic
verification, and we then report on experimental results in
\Cref{sec:exp}. We conclude in \Cref{sec:conc}. Some proofs and
experimental results are deferred to the appendix of the full version~\cite{BCCCES25}.

\section{Weakly acyclic languages}
\label{sec:wa}
We assume familiarity with basic automata theory. Let us recall some
notions. Let $\A = (Q, \Sigma, \delta, q_0, F)$ be a deterministic
finite automaton (DFA). We write $\lang{\A}$ to denote the language
accepted by $\A$. The language accepted by a state $q \in Q$, denoted
$\lang{q}$, is the language accepted by the DFA $(Q, \Sigma, \delta,
q, F)$. We write $\Sigma^+$ to denote the set of all nonempty words
over alphabet $\Sigma$.

For every word $w \in \Sigma^*$, we define $\alp{w}$ as the set of
letters that occur within $w$, \eg\ $\alp{baacac} = \{a, b, c\}$. The
\emph{residual} of a language $L \in \Sigma^*$ with respect to a word
$u \in \Sigma^*$ is $L^u \defeq \{v \in \Sigma^* : uv \in L\}$. Recall
that a language is regular iff it has finitely many
residuals. Moreover, every regular language has a unique minimal
DFA. The states of a minimial DFA $\A$ accept distinct residuals of
$\lang{\A}$, and each residual of $\lang{\A}$ is accepted by a state
of $\A$. For example, let $L \subseteq \{a, b\}^*$ be the language
described by $(a b^* + b + \ew)a b^*$. The minimal DFA for $L$ is
depicted in \Cref{fig:aut}. Its residuals are $L^\ew$, $L^a$, $L^b$,
$L^{aa}$ and $L^{aaa}$.

\begin{figure}
  \centering
  \begin{tikzpicture}[
  node distance=2.5cm, auto, thick, initial text={},
  scale=0.8, transform shape
]
  \node[state, initial]     (q0) {};
  \node[state, right of=q0, accepting] (q1) {};
  \node[state, right of=q1, accepting] (q2) {};
  \node[state, right of=q2] (q3) {};
  \node[state, below=0.25cm of q1] (q4) {};
  
  \path[->]
  (q0) edge node {$a$} (q1)
  (q1) edge node {$a$} (q2)
  (q2) edge node {$a$} (q3)

  (q0) edge node[swap] {$b$} (q4)
  (q4) edge node[swap] {$a$} (q2)
  
  (q1) edge[loop above] node {$b$} ()
  (q2) edge[loop above] node {$b$} ()
  (q3) edge[loop above] node {$a, b$} ()

  (q4) edge[bend right=10] node[swap] {$b$} (q3)
  ;
\end{tikzpicture}
  \caption{Example of a minimal DFA (whose language is weakly
    acyclic).}\label{fig:aut}
\end{figure}
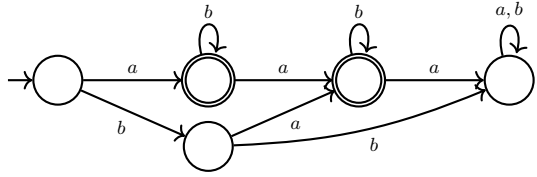

We say that a DFA $\A = (Q, \Sigma, \delta, q_0, F)$ is \emph{weakly
acyclic} if for every $q \in Q$, $w \in \Sigma^+$ and $a
\in \alp{w}$, it is the case that $\delta(q, w) = q$ implies
$\delta(q, a) = q$. For example, \Cref{fig:aut} depicts a weakly
acyclic DFA. Equivalent definitions of weak acyclicity are as follows:
\begin{itemize}
\item the binary relation $\preceq$ over $Q$ given by ``$q \preceq q'$
  if $\delta(q, w) = q'$ for some word $w$'' is a partial order;
  
\item each strongly connected component of the underlying directed
  graph of $\A$ contains a single state; and

\item the underlying directed graph of $\A$ does not have any
  cycle beyond self-loops.
\end{itemize}

Analogously, a nondeterministic finite automaton (NFA) $\A = (Q,
\Sigma, \delta, q_0, F)$ is \emph{weakly acyclic} if $q \in \delta(q,
w)$ implies $\delta(q, a) = \{q\}$ for all $a \in \alp{w}$. In other
words, the underlying directed graph does not contain any simple cycle
beyond self-loops, and nondeterminism with a letter $a$ can only occur
from states without self-loops labeled by $a$.

We make some observations whose proofs appear in the appendix of the
long version of this work. Let us observe that the powerset
construction (determinization), and DFA minimization, both preserve
weak acyclicity.

\begin{restatable}{proposition}{propAcycPow}\label{prop:acyc:pow}
  Applying the powerset construction to a weakly acyclic NFA yields a
  weakly acyclic DFA.
\end{restatable}

\begin{restatable}{proposition}{propAcycMinDfa}\label{prop:acyc:min:dfa}
  Let $\A$ be a weakly acyclic DFA. The minimal DFA that accepts
  $\lang{\A}$ is also weakly acyclic.
\end{restatable}

We say that a language is \emph{weakly acyclic} if it is accepted by a
weakly acyclic automaton $\A$. By the above propositions, this means
that $\A$ can interchangeably be an NFA, a DFA, or a minimal DFA. The
lemma below yields a characterization of weakly acyclic languages in
terms of residuals rather than automata.

\begin{restatable}{lemma}{lemResiduals}\label{lem:residuals}
  A regular language $L \subseteq \Sigma^*$ is weakly acyclic iff for
  all $u \in \Sigma^*$ and $v \in \Sigma^+$ it is the case that $L^u =
  L^{uv}$ implies $L^u = L^{uw}$ for every $w \in \alp{v}^*$.
\end{restatable}

It can also be shown that weakly acyclic languages are precisely the
languages described by such \emph{weakly acyclic expressions}: $r
::= \emptyset \mid \Gamma^* \mid \Lambda^* a r \mid r + r$ where
$\Gamma, \Lambda \subseteq \Sigma$ and
$a \in \Sigma \setminus \Lambda$ (see the appendix of the full
version). Note that these expressions are essentially
``$\mathcal{R}$-expressions'', which are known to be equivalent to
weakly acyclic languages via algebraic automata
theory. Moreover, \cite[Thm.~7]{KMT17}, shows that weaky acyclic DFAs
and NFAs are equivalent, and the proof is very similar our conversion
from weakly acyclic NFAs to expressions given in the appendix of the
full version.

Complementing a weakly acyclic DFA preserves weak acyclicity since no
cycle is created. Moreover, weakly acyclic regular expressions allow
for union. Thus:

\begin{proposition}\label{prop:wa:closures}
  Weakly acyclic languages are closed under complementation, union and
  intersection.
\end{proposition}

\section{A data structure for weakly acyclic languages}
\label{sec:ds}
Note that if a language $L$ is weakly acyclic, then so is $L^a$ for
all $a \in \Sigma$. From this simple observation, one can imagine an
infinite deterministic automaton where each state is a weakly acyclic
language $L$, and each $a$-transition leads to $L^a$. Let us define
this ``master automaton'' formally.

\begin{definition}
  The \emph{master automaton} over alphabet $\Sigma$ is the tuple $\M
  = (Q_{\M}, \Sigma, \allowbreak \delta_{\M}, F_{\M})$, where
  \begin{itemize}
  \item $Q_{\M}$ is is the set of all weakly acyclic languages over
    $\Sigma$;

  \item $\delta_{\M} \colon Q_{\M} \times \Sigma \rightarrow Q_{\M}$
    is given by $\delta_{\M}(L, a) \defeq L^a$ for all $L \in Q_{\M}$
    and $a \in \Sigma$;

  \item $L \in F_{\M}$ iff $\ew \in L$.
  \end{itemize}
\end{definition}

Given weakly acyclic languages $K$ and $L$, let $K \preceq L$ denote
that $K = L^w$ for some word $w$. An immediate consequence of the
definition of weakly acyclic DFAs is that ${\preceq}$ is a partial
order. The minimal elements of ${\preceq}$ satisfy the equation $L =
L^a$ for every letter $a$. The equation has exactly two solutions: $L
= \emptyset$ and $L = \Sigma^*$. Moreover, we can easily show by
structural induction that ${\preceq}$ has no infinite descending
chains. This allows to reason about a weakly acyclic language
recursively through its residuals until reaching $\emptyset$ or
$\Sigma^*$.

\begin{restatable}{proposition}{propFundOne}\label{prop:fund1}
  Let $L$ be a weakly acyclic language. The language accepted from the
  state $L$ of the master automaton is $L$.
\end{restatable}

By \Cref{prop:fund1}, we can look at the master automaton as a
structure containing DFAs recognizing all the weakly acyclic
languages. To make this precise, each weakly acyclic language $L$
determines a DFA $\A_L = (Q_L, \Sigma, \delta_L, q_{0L}, F_L)$ as
follows: $Q_L$ is the set of states of the master automaton reachable
from the state $L$, $q_{0L}$ is the state $L$, $\delta_L$ is the
projection of $\delta_{\M}$ onto $Q_L$, and $F_L \defeq F_{\M} \cap
Q_L$. It is easy to show that $\A_L$ is the \emph{minimal} DFA
recognizing $L$.

\begin{restatable}{proposition}{propFundTwo}\label{prop:fund2}
  For every weakly acyclic language $L$, automaton $\A_L$ is the
  minimal DFA recognizing $L$.
\end{restatable}

\Cref{prop:fund2} allows us to define a data structure for
representing finite sets of weakly acyclic languages. Loosely
speaking, the structure representing the languages $\LL = \{L_1,
\ldots, L_n\}$ is the fragment of the master automaton containing the
states recognizing $L_1, \ldots, L_n$ and their descendants. It is a
DFA with multiple initial states which we call the \emph{weakly
acyclic diagram for $\LL$}. Formally:

\begin{definition}
  Let $\LL = \{L_1, \ldots, L_n\}$ be weakly acyclic languages over
  $\Sigma$. The \emph{weakly acyclic diagram} $\A_\LL$ is the tuple
  $(Q_\LL, \Sigma, \delta_\LL, Q_{0\LL}, F_\LL)$ where $Q_\LL$ is the
  subset of states of the master automaton reachable from at least one
  of $L_1, \ldots, L_n$; $Q_{0\LL} \defeq \{L_1, \ldots, L_n\}$;
  $\delta_\LL$ is the projection of $\delta_M$ onto $Q_\LL$; and
  $F_\LL = F_\M \cap Q_\LL$.
\end{definition}

In order to manipulate weakly acyclic diagrams, we represent them as
tables of nodes. Let us fix $\Sigma = \{a_1, a_2, \ldots, a_m\}$. A
\emph{node} is a triple $q \colon (s, b)$ where
\begin{itemize}
\item $q$ is the node \emph{identifier};

\item $s = (q_1, q_2, \ldots, q_m)$ is the \emph{successor tuple} of
  the node, where each $q_i$ is either a node identifier or the
  special symbol \self; and
  
\item $b \in \{0, 1\}$ is the \emph{acceptance flag} that indicates
  whether $q$ is accepting or not.
\end{itemize}

We write $q^{a_i}$ to denote $q_i$. We write $\lang{s, b}$ for the
language defined recursively as follows, where $i \in [1..m]$ and $w
\in \Sigma^*$:
\begin{align*}
  \ew \in \lang{q} &\iff b = 1, \\
  a_i w \in \lang{q} &\iff \left((q^{a_i} = \self \land w \in \lang{q}) \lor
  (q^{a_i} \neq \self \land w \in \lang{q^{a_i}})\right).
\end{align*}
For every node $q \colon (s, b)$, we write $\lang{q}$ to denote
$\lang{s, b}$.

We denote the identifiers of the nodes for languages $\emptyset$ and
$\Sigma^*$ by $\qnone$ and $\qall$, respectively, \ie, $\qnone \colon
((\self, \ldots, \self), 0)$ and $\qall \colon ((\self, \ldots,
\self), 1)$.

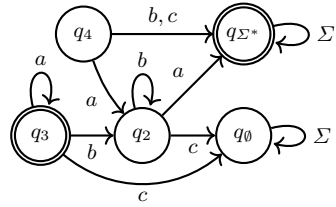
\begin{figure}
  \hfill
  \begin{minipage}[t]{0.45\textwidth}
    \vspace{0pt} %% Hack for top alignment
    %\resizebox{0.7\columnwidth}{!}{%
    \begin{tabular}{ccccc}
      \toprule

      \textbf{identifier} &
      \multicolumn{3}{c}{\textbf{succ.\ tuple}} &
      \textbf{flag} \\

      & $a$ & $b$ & $c$ & \\

      \midrule

      $q_4$ & $q_2$ & $\qall$ & $\qall$ & $0$ \\

      $q_3$ & $\self$ & $q_2$ & $\qnone$ & $1$ \\

      $q_2$ & $\qall$ & $\self$ & $\qnone$ & $0$ \\

      $\qall$ & $\self$ & $\self$ & $\self$ & $1$ \\

      $\qnone$ & $\self$ & $\self$ & $\self$ & $0$ \\

      \bottomrule
    \end{tabular}%}
  \end{minipage}%
  \begin{minipage}[t]{0.45\textwidth}
    \vspace{2pt}  %% Hack for top alignment
    \begin{tikzpicture}[auto, node distance=1.5cm, thick, scale=0.9, transform shape]
  \node[state, accepting]              (q1) {$\qall$};
  \node[state,            below of=q1] (q0) {$\qnone$};
  \node[state,            left  of=q0] (q2) {$q_2$};
  \node[state, accepting, left  of=q2] (q3) {$q_3$};
  \node[state,            left=1.5cm of q1] (q4) {$q_4$};

  \path[->]
  (q1) edge[loop right]    node[]     {$\Sigma$} ()
  (q0) edge[loop right]    node[]     {$\Sigma$} ()
  (q3) edge[loop above]    node[]     {$a$}      ()
  (q3) edge[]              node[swap] {$b$}      (q2)
  (q3) edge[bend right=40] node[swap] {$c$}      (q0)
  (q2) edge[loop above]    node       {$b$}      ()
  (q2) edge[]              node[]     {$a$}      (q1)
  (q2) edge[]              node[swap] {$c$}      (q0)
  (q4) edge[]              node[]     {$b, c$}   (q1)
  (q4) edge[bend right=10] node[swap] {$a$}      (q2)
  ;
\end{tikzpicture}
  \end{minipage}
  \hfill
  \caption{Example of the representation of weakly acyclic languages.}\label{fig:ex:ds}
\end{figure}

\begin{example}
  Let $K$ and $L$ be the languages over $\Sigma = \{a, b, c\}$
  described by the regular expressions $a^* (\ew + b^+ a \Sigma^*)$
  and $a b^* a \Sigma^* + (b + c)\Sigma^*$. These languages are represented by
  the table depicted on the left of \Cref{fig:ex:ds}, which
  corresponds to the diagram on the right. It is readily seen that
  $\lang{q_3} = K$ and $\lang{q_4} = L$. \qed
\end{example}

Note that there is a risk that two distinct nodes represent the same
language. For example, suppose $\Sigma = \{a, b\}$ and that we create
a node $p \colon ((\self, \qall), 1)$. We obtain the following table,
where $\lang{p} = \lang{\qall} = \{a, b\}^*$:

\begin{center}
  \resizebox{0.325\columnwidth}{!}{%
    \begin{tabular}{cccc}
      \toprule
      
      \textbf{identifier} &
      \multicolumn{2}{c}{\textbf{succ.\ tuple}} &
      \textbf{flag} \\
      
      & $a$ & $b$ & \\
      
      \midrule
      
      $p$ & $\self$ & $\qall$ & $1$ \\
      
      $\qall$ & $\self$ & $\self$ & $1$ \\
      
      \bottomrule
    \end{tabular}
  }
\end{center}

We will avoid this: we will maintain a table of nodes such that the
language of distinct nodes is distinct. The procedure \make of
\Cref{alg:make} serves this purpose. Given a successor tuple $s$ and
an acceptance flag $b$, it first checks whether there is already an
entry for language $\lang{s, b}$, and if not it creates one with a new
identifier. The identifiers created by \make are generated in
increasing order. We assume that the table initially contains $\qnone$
and $\qall$ as the first two identifiers.

\begin{algorithm}
  \DontPrintSemicolon
  \LinesNumbered
  \SetKwProg{myproc}{}{}{}
  \KwIn{$s = (q_1, \ldots, q_m)$ and $b \in \{0, 1\}$ where $q_i = \self$ or an existing identifier}
  \KwResult{unique identifier $q$ with language $\lang{s, b}$ and $q \notin s$}
  \algskip

  \myproc{\make{$s, b$}\algdelim}{
    \lIf{the table contains $q \colon (s, b)$}{
      \Return{$q$}\label{ln:make:exists}
    }
    \Else{
      %% \For{$q' \in s$}{
      %%   $s' \leftarrow s[q' / \self]$\;
      %%   \algskip
      %%
      %%   \lIf{the table contains $q \colon (s', b)$}{
      %%     \Return{$q$}\label{ln:make:equiv}
      %%   }
      %% }
      %%
      $q' \leftarrow \max\{r \in s : r \neq \self\}$\;
      $s' \leftarrow s[q' / \self]$\;
      \algskip
      \algskip

      \lIf{the table contains $q \colon (s', b)$}{
        \Return{$q$}\label{ln:make:equiv}
      }
      \Else{
        $q \leftarrow \text{next fresh identifier}$\;
        \textbf{add} $q \colon (s, b)$ to the table\;
        \Return{$q$}\;
      }
    }
  }
  
  \caption{Algorithm to maintain nodes.}\label{alg:make}
\end{algorithm}

Let us explain how $\make$ checks whether a node must be
created. Given a tuple $s = (q_1, q_2, \ldots, q_m)$, we write $s[r /
r']$ to denote the tuple obtained from $s$ by replacing each
occurrence $q_i = r$ by $r'$. As already explained, when trying to add
$(s, b)$ to the table, there might already be an entry $q \colon (s',
b')$ with $\lang{s', b'} = \lang{s, b}$. It can be shown that this
happens iff $b' = b$ and there exists $q' \in s$ such that $s' = s[q'
/ \self]$ (see the proof of \Cref{prop:make} below in the full
version~\cite{BCCCES25}). In fact, it can be shown that such a
$q'$ must be the maximal identifier of $s$; hence, there is no need to
check all elements of $s$. For example, in the aforementioned
problematic case, we have $s = (\self, \qall)$ and $s[\qall / \self] =
(\self,
\self)$ which is already in the table, so no node is created.

The amortized time of \make belongs to $\O(|\Sigma|)$. Further, it
works as intended:

\begin{restatable}{proposition}{propMake}\label{prop:make}
  A table obtained from successive calls to \Cref{alg:make} does not
  contain nodes $q \neq q'$ such that $\lang{q} = \lang{q'}$.
\end{restatable}

\section{Operations on weakly acyclic languages}
\label{sec:oper}
\subsection{Unary and binary operations}\label{ssec:binop}

\Cref{alg:comp,alg:inter} describe recursive procedures that
respectively complement a weakly acyclic language, and intersect two
weakly acyclic languages. They are \emph{very similar} to the
classical algorithms for negation and conjunction in OBDDs
(see~\cite[Chap.~6, Algorithms~27 and~26]{EB23}); they essentially
only differ in the usage of \self. For this reason, their correcteness
and time complexity are established with the same arguments. Let us
expand for readers less familiar with OBDDs.

Procedure \comp terminates because on each recursive call, $q^a$ is a
smaller identifier than $q$. Procedure \inter terminates because on
each recursive call, either $p^a \neq \self$ and $p^a$ is a smaller
identifier than $p$, or likewise for $q^a$ and $q$.

\begin{center}
  \resizebox{0.975\columnwidth}{!}{%
    \begin{minipage}[t]{0.47\textwidth}
      \begin{algorithm}[H]
        \DontPrintSemicolon
        \LinesNumbered
        \SetKwProg{myproc}{}{}{}
        \KwIn{identifier $q$}
        \KwResult{id.\ $r$ : $\lang{r} = \overline{\lang{q}}$}
        \algskip
        
        \myproc{\comp{$q$}\algdelim}{
          \If{$q = \qnone$}{
            \Return{$\qall$}\;
          }
          \ElseIf{$q = \qall$}{
            \Return{$\qnone$}\;
          }
          \Else{
            \For{$a \in \Sigma$\label{ln:comp:for}}{
              \If{$q^a = \self$}{
                $s_a \leftarrow \self$\;
              }
              \Else{
                $s_a \leftarrow \comp{$q^a$}$\;
              }
            }
            
            \algskip
            \algskip
            \Return{\make{$s, \neg \flag{q}$}}\label{ln:comp:return}\hspace*{-1cm}
          }
        }
        
        \caption{Complement.}\label{alg:comp}
      \end{algorithm}
    \end{minipage}%
    \begin{minipage}[t]{0.55\textwidth}
      \begin{algorithm}[H]
        \DontPrintSemicolon
        \LinesNumbered
        \SetKwProg{myproc}{}{}{}
        \KwIn{identifiers $p$ and $q$}
        \KwResult{id.\ $r$ : $\lang{r} = \lang{p} \cap \lang{q}$}
        \algskip
        
        \myproc{\inter{$p, q$}\algdelim}{
          \If{$p = \qnone$ or $q = \qnone$}{
            \Return{$\qnone$}\;
          }
          \ElseIf{$p = \qall$}{
            \Return{$q$}\;
          }
          \ElseIf{$q = \qall$}{
            \Return{$p$}\;
          }
          \Else{
            \For{$a \in \Sigma$\label{ln:inter:for}}{
              $\mathmakebox[32pt][l]{p' \leftarrow p^a} \textbf{ if } \mathmakebox[42pt][l]{p^a \neq \self} \textbf{ else } p$\hspace*{-2pt}\;
              $\mathmakebox[32pt][l]{q' \leftarrow q^a} \textbf{ if } \mathmakebox[42pt][l]{q^a \neq \self} \textbf{ else } q$\hspace*{-2pt}\;
              \algskip
              \algskip
              
              \If{$p' = p$ and $q' = q$}{
                $s_a \leftarrow \self$\;
              }
              \Else{
                $s_a \leftarrow \inter{$p', q'$}$\;
              }
            }
            
            \algskip
            \Return{\make{$s, \flag{p} \land \flag{q}$}}\label{ln:inter:return}\hspace*{-1cm}
          }
        }
        
        \caption{Intersection.}\label{alg:inter}
      \end{algorithm}
  \end{minipage}}
\end{center}

The two procedures may have an exponential complexity since the
recursion may recompute the same values repeatedly. However, this is
easily avoided with memoization. This will apply to all our
algorithms, \ie, each recursive procedure $\texttt{t}(x)$ can first
check in a dictionary whether the output value for $x$ has already
been computed, and only proceed to compute it if not. To avoid
cluttering the presentation, we will not explicitly describe such
memoization in our pseudocode.

The following proposition is established with standard arguments (see
the appendix of the full version).

\begin{proposition}\label{prop:comp:inter:full}
  \Cref{alg:comp} and \Cref{alg:inter} are correct. Moreover, with
  memoization, they respectively work in time $\O(|\Sigma| n_q)$ and
  $\O(|\Sigma| n_p n_q)$, where $n_x$ is the number of nodes reachable
  from node $x$.
\end{proposition}

It is easy to adapt \Cref{alg:inter} to other binary operations such
as union: It suffices to change the base cases and the way the
acceptance flag is set.

Note that testing emptiness, universality and language equality can be
done in \emph{constant time}. Indeed, as the table keeps a unique
identifier for each language, we have $\lang{q} = \emptyset$ iff $q =
\qnone$, $\lang{q} = \Sigma^*$ iff $q = \qall$, and $\lang{p} =
\lang{q}$ iff $p = q$.

\subsection{Fixed-length relations: a general approach}\label{ssec:rel}

For a relation $R \subseteq \Sigma^* \times \Sigma^*$ and a language
$L \subseteq \Sigma^*$, we define $\Post{R}{L} \defeq \{v \in \Sigma^*
: u \in L, (u, v) \in R\}$ and $\Pre{R}{L} \defeq \{u \in \Sigma^* :
(u, v) \in R, v \in L\}$. A relation is said to be \emph{fixed-length}
if $(u, v) \in R$ implies $|u| = |v|$.

We consider fixed-length regular relations, \ie, those that can be
represented by automata over alphabet $\Sigma \times \Sigma$, which we
call transducers. Given a fixed-length regular relation $R$ and a
weakly acyclic language $L$, $\Post{R}{L}$ and $\Pre{R}{L}$ may not be
weakly acyclic. For example, consider the relation $R$ and the
language $L$ depicted by the transducer and automata of
\Cref{fig:rel:bad}. We have $\Post{R}{L} = (a + b)^* b$ which is not
weakly acyclic (as its minimal DFA has a nontrivial cycle).

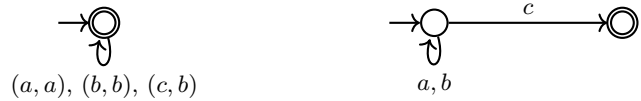
\begin{figure}
  \centering
  \begin{tikzpicture}[
  node distance=2.5cm, auto, thick, initial text={},
  scale=1.0, transform shape
]
  \tikzstyle{cstate} = [state, minimum size=10pt]
  
  %% T
  \node[cstate, initial, accepting] (p) {};
  
  \path[->]
  (p) edge[loop below] node {$(a, a)$, $(b, b)$, $(c, b)$} ()
  ;

  %% A
  \node[cstate, initial, right=4cm of p] (q) {};
  \node[cstate, accepting, right of=q] (r) {};
  
  \path[->]
  (q) edge[loop below] node {$a, b$} ()
  (q) edge             node {$c$}    (r)
  ;    
\end{tikzpicture}  
  \caption{\emph{Left}: A transducer $\T$ that converts each
    occurrence of letter $c$ into letter $b$. \emph{Right}: A weakly
    acyclic DFA $\A$ accepting language $(a + b)^*
    c$.}\label{fig:rel:bad}
\end{figure}

Yet, under the guarantee that the resulting language is weakly
acyclic, we can compute $\Post{R}{L}$ and $\Pre{R}{L}$. The key
observation is that if a DFA accepting a weakly acyclic language has a
cycle, then this cycle can be ``contracted''. Thus, given a transducer
$\T$ and a node $q$, we apply the powerset contruction on the pairing
of $T$ and $q$, and we contract cycles on the fly. Informally, cycle
detection is achieved by testing whether a successor is already on the
call stack. \Cref{alg:pre:gen} implements this idea for
$\Pre{R}{L}$. In the pseudocode, symbols $S$ and $S'$ denote sets of
pairs of states.

\begin{algorithm}
  \DontPrintSemicolon
  \LinesNumbered
  \SetKwFunction{make}{make}
  \SetKwFunction{pre}{pre}
  \SetKwFunction{prep}{pre-aux}
  \SetKwFunction{tsucc}{succ}
  \SetKwProg{myproc}{}{}{}

  \KwIn{transducer $\T = (P, \Sigma \times \Sigma, \delta, p_0, F)$, and identifier $q$ such that $\Pre{\lang{T}}{\lang{q}}$ is weakly acyclic}
  \KwResult{identifier $r$ such that $\lang{r} = \Pre{\lang{\T}}{\lang{q}}$}
  \algskip

  \myproc{\pre{$\T, q$}\algdelim}{
    \myproc{\tsucc{$q, b$}\algdelim}{
      \lIf{$q^b \neq \self$}{
        \Return{$q^b$}
      }
      \lElse{
        \Return{$q$}
      }
    }

    \algskip
    \algskip
    \myproc{\prep{$S$}\algdelim}{
      \textbf{mark} $S$\;
      
      \algskip
      \algskip
      \For{$a \in \Sigma$}{
        $S' \leftarrow \bigcup_{b \in \Sigma}
        \{(p', q') : (p, q) \in S,\ p' \in \delta(p, (a, b)),\ q' = \tsucc{$q, b$}\}$\label{ln:S:def}
        
        \algskip
        \algskip
        \lIf{$S'$ is marked}{
          $s_a \leftarrow \self$\tcp*[f]{Contract detected cycle}\label{ln:S:marked}
        }
        \lElse{
          \hspace*{69.25pt}$s_a \leftarrow \prep{$S'$}$\label{ln:S:unmarked}
        }
      }
      
      \algskip
      \algskip
      \textbf{unmark} $S$\;
      
      \algskip
      \Return{\make{$s,\ \exists (p, q) \in S : p \in F \land \flag{q}$}}\label{ln:pre:gen:make}
    }

    \algskip
    \algskip
    \Return{\prep{$\{(p_0, q)\}$}}\;
  }
  
  \caption{General algorithm for the Pre operation}\label{alg:pre:gen}
\end{algorithm}

Let us analyze \Cref{alg:pre:gen}. We write $\Pre{p}{q}$ to denote
$\Pre{\lang{p}}{\lang{q}}$. Let $\lang{S} \defeq \bigcup_{(p, q) \in
  S} \Pre{p}{q}$. Given $S$ and $a \in \Sigma$, we write $S \trans{a}
S'$ where
\[
S' \defeq \bigcup_{b \in \Sigma} \left\{(p', q') : (p, q) \in S, p' \in
\delta(p, (a, b)), q' = \tsucc{$q, b$}\right\}.
\]
In words, ``$S \trans{a} S'$'' denotes that \prep{$S$} constructs the
set $S'$ on Line~\ref{ln:S:def}. By extension, we write $S \trans{\ew}
S'$ if $S = S'$, and we write $S \trans{a_1 a_2 \cdots a_n} S'$ if $S
= S_0 \trans{a_1} S_1 \trans{a_2} S_2 \cdots \trans{a_n} S_n = S'$ for
some $S_0, S_1, \ldots, S_n$.

\begin{restatable}{proposition}{propSRes}\label{prop:s:res}
  Let $w \in \Sigma^*$. If $S \trans{w} S'$, then $\lang{S}^w =
  \lang{S'}$.
\end{restatable}

\begin{restatable}{proposition}{propSCycle}\label{prop:S:cycle}
  If $S_0 \trans{*} S' \trans{*} S \trans{a} S'$ for some $a \in
  \Sigma$, and $\lang{S_0}$ is weakly acyclic, then $\lang{S} =
  \lang{S'}$.
\end{restatable}

\begin{proposition}
  \Cref{alg:pre:gen} is correct and terminates.
\end{proposition}

\begin{proof}
  Termination follows easily by the fact that only finitely many
  subsets can be generated (see the appendix of the full version). Let
  us prove correctness.
  
  Let $S_0 \defeq \{(p_0, q)\}$ be the input and let \prep{$S$} be a
  recursive call made by \pre{$S_0$}. Recall that $\lang{S} =
  \bigcup_{(p, q) \in S} \Pre{p}{q}$. So, we must show that
  $\lang{\prep{$S$}} = \lang{S}$. To do so, we must prove that, for
  each $a \in \Sigma$, \Cref{alg:pre:gen} assigns to $s_a$ the state
  accepting $\lang{S}^a$. Moreover, we must show that the acceptance
  flag is set correctly. We show the latter first, and then the
  former.

  We have $\ew \in \bigcup_{(p, q) \in S} \Pre{p}{q}$ iff there exists
  $(p, q) \in S$ such that $(\ew, \ew) \in \lang{p}$ and $\ew \in
  \lang{q}$ iff there exists $(p, q) \in S$ such that $p \in F$ and
  $\flag{q} = \true$. Thus, the flag is set correctly on
  Line~\ref{ln:pre:gen:make}.

  Let $a \in \Sigma$. Let $S'$ be the set computed on Line~\ref{ln:S:def}
  for letter $a$. By definition, we have $S \trans{a} S'$. We make a
  case distinction on whether $S'$ is marked or not.

  \medskip
  \noindent\emph{Case 1: $S'$ is marked}. Since $S'$ is marked, we
  have $S_0 \trans{*} S' \trans{*} S$. Therefore, $S_0 \trans{*} S'
  \trans{*} S \trans{a} S'$. By \Cref{prop:S:cycle}, we have $\lang{S}
  = \lang{S'}$. Moreover, by \Cref{prop:s:res}, we have $\lang{S}^a =
  \lang{S'}$. Thus, $\lang{S}^a = \lang{S'} = \lang{S}$. As $S'$ is
  marked, \Cref{alg:pre:gen} executes Line~\ref{ln:S:marked}, which
  correctly assigns ``$s_a \leftarrow \self$''.

  \medskip
  \noindent\emph{Case 2: $S'$ is unmarked}. By \Cref{prop:s:res},
  $\lang{S}^a = \lang{S'}$. As $S'$ is unmarked, \Cref{alg:pre:gen}
  executes Line~\ref{ln:S:unmarked}, which correctly assigns ``$s_a
  \leftarrow \prep{$S'$}$''. \qed
\end{proof}

Note that \Cref{alg:pre:gen} has exponential worse-case time
complexity since this is already the case for the particular case of
OBDDs (see \cite[Chap.~6.5, pp.~148--149]{EB23}). Indeed, given an arbitrary NFA $A$, one can find a
(deterministic) transducer and a DFA whose composition yields $A$, and it is well
known that determinizing an NFA  blows up its size exponentially in the worst case.

\subsection{Fixed-length relations: a more efficient approach}\label{ssec:wa:rel}

As just mentioned, pairing a transducer and an automaton can result in
a nondeterministic automaton. For this reason, \Cref{alg:pre:gen}
determinizes on the fly by constructing subsets of states, and
minimizes on the fly by contracting cycles.

We describe a setting in which determinism is guaranteed, and hence
where cycle contraction is avoided. This allows for a polynomial-time
time procedure.

Observe that a weakly acyclic language $K$ over $\Sigma \times \Sigma$
can be represented with our data structure. As shown in
\Cref{fig:rel:bad}, even if $K$ and $L$ are weakly acyclic, it is not
necessarily the case that $\Post{K}{L}$ and $\Pre{K}{L}$ are weakly
acyclic. Yet, we provide a sufficient condition under which weak
acyclicity is guaranteed.

Let $K$ and $L$ be weakly acyclic languages respectively over $\Sigma
\times \Sigma$ and $\Sigma$. We say that $K$ and $L$ are
\emph{pre-compatible} if for every $a \in \Sigma$ the following holds:
\begin{itemize}
\item there exists at most one $b \in \Sigma$ such that $K^{(a, b)}
  \neq \emptyset$ and $L^b \neq \emptyset$; and

\item if $K^{(a, b)} \neq \emptyset$ and $L^b \neq \emptyset$, then
  $K^{(a, b)}$ and $L^b$ are pre-compatible.
\end{itemize}

\Cref{alg:pre:comp} describes a procedure that computes a node
accepting $\Pre{K}{L}$. Similarly to \Cref{alg:comp,alg:inter}, a
simple analysis yields the following proposition (see the appendix of
the full version).

\begin{proposition}\label{prop:pre:comp:full}
  \Cref{alg:pre:comp} is correct. Moreover, with memoization, it works
  in time $\O(|\Sigma|^2 n_p n_q)$, where $n_x$ is the number of nodes
  reachable from node $x$.
\end{proposition}

\begin{algorithm}
  \DontPrintSemicolon
  \LinesNumbered
  \SetKwFunction{make}{make}
  \SetKwFunction{pre}{pre}
  \SetKwProg{myproc}{}{}{}
  \KwIn{identifier $p$ over $\Sigma \times \Sigma$, identifier $q$ over $\Sigma$ such that $\lang{p}$ and $\lang{q}$ are pre-compatible}
  \KwResult{identifier $r$ such that $\lang{r} = \Pre{\lang{p}}{\lang{q}}$}
  \algskip

  \myproc{\pre{$p, q$}\algdelim}{
    \For{$a \in \Sigma$\label{ln:pre:comp:for}}{
      \lIf{$\neg(\exists b \in \Sigma : p^{(a, b)} \neq p_\emptyset \land q^b \neq q_\emptyset)$}{
        $s_a \leftarrow q_\emptyset$
      }
      \Else{
        \textbf{let} $b$ be the unique letter such that $p^{(a, b)} \neq p_\emptyset$ and $q^b \neq q_\emptyset$

        \algskip
        $\mathmakebox[42pt][l]{p' \leftarrow p^{(a, b)}} \textbf{ if } \mathmakebox[52pt][l]{p^{(a, b)} \neq \self} \textbf{ else } p$\;

        $\mathmakebox[42pt][l]{q' \leftarrow q^b} \textbf{ if } \mathmakebox[52pt][l]{q^b \neq \self} \textbf{ else } q$\;

        \algskip
        \algskip
        \lIf{$p' = p \land q' = q$\label{ln:pp:qq}}{
          $s_a \leftarrow \self$
        }
        \lElse{
          \hspace*{75.65pt}$s_a \leftarrow \pre{$p', q'$}$
        }
      }
    }
    
    \algskip
    \Return{\make{$s, \flag{p} \land \flag{q}$}}\label{ln:pre:return}
  }
  
  \caption{Algorithm for the Pre operation under pre-compatibility}\label{alg:pre:comp}
\end{algorithm}

\section{Weakly acyclic (regular) model checking}
\label{sec:model:check}
In this section, we present interesting examples of systems that can
be modeled with WADs. In our setting, a system has a set of
configurations equipped with a partial order ${\preceq}$ and a
monotone labeled transition relation, \ie, if $C \trans{t} D$ and $C'
\succeq C$, then $C' \trans{t} D'$ for some $D' \succeq
C'$. The \emph{safety verification problem} asks, given configurations
$C$ and $C'$, whether $C \trans{*} C''$ for some $C''
\succeq C'$.

For every $D$, let $\cPre{D} \defeq \{C : C \trans{*} D\}$. We extend
this notation to sets, \ie, $\cPre{X} \defeq
\bigcup_{C \in X} \cPre{C}$. Let $\cPreI{*}{X} \defeq \bigcup_{i \geq
  0} \cPreI{i}{X}$. A set of configurations $X$ is
\emph{upward-closed} if $\upc{X} = X$, where $\upc{X} \defeq \{C' : C'
\succeq C \text{ and } C \in X\}$.

Safety verification amounts to testing whether $C \in
\cPreI{*}{\upc{C'}}$. If the system is well structured, namely if
${\preceq}$ is a well-quasi-order, then it is well-known that
$\cPreI{*}{\upc{C'}} = \cPreI{i}{\upc{C'}}$ for some $i \geq 0$. Thus,
verification can be done by the \emph{backward algorithm}: compute
$\cPreI{0}{\upc{C'}} \cup \cPreI{1}{\upc{C'}} \cup \cdots$ until
reaching a fixed point. So, for well-structured transition systems,
the approach is as follows:
\begin{itemize}
\item Choose a suitable alphabet $\Sigma$;
  
\item For each configuration $D$, define an encoding $\cEnc{\upc{D}}
  \subseteq \Sigma^*$ of $\upc{D}$ which is a weakly acyclic language;

\item For each $t$, define a transducer $\T_t$ for the language
  $\{(\cEnc{D}, \cEnc{D'}) : D \trans{t} D'\} \subseteq (\Sigma \times
  \Sigma)^*$;

\item Create a weakly acyclic diagram for $\lang{q_0} =
  \cEnc{\upc{C'}}$ and compute $\lang{q_{i+1}} =
  \Pre{\T_t}{\lang{q_i}}$ iteratively until $q_{i+1} = q_i$.
\end{itemize}

Since ${\preceq}$ is a well-quasi-order, any upward-closed set $X$ of
configurations has a finite basis, \ie, there exist configurations
$C_1, \ldots, C_n$ such that $X = \bigcup_{i \in [1..n]}
\upc{C_i}$. Thus, as weakly acyclic languages are closed under union,
$\cEnc{X}$ must be weakly acyclic. Moreover, since $\cPre{X}$ is
upward-closed by monotonicity, $\Pre{\T_t}{\cEnc{X}}$ is necessarily
weakly acyclic.
  
\subsection{Lossy channel systems}

A \emph{channel system} is a finite directed graph whose nodes $P$ are
called \emph{states}, and whose arcs, called \emph{transitions}, are
labeled by $\lcsread{a}{i}$, $\lcswrite{a}{i}$ or $\lcsnop{}$ with $a
\in \Gamma$ and $i \in [1..k]$. There are $m$ processes starting in
some states, and evolving by reading from the left and writing to the
right of channels, and by moving to other states. A \emph{lossy
channel system (LCS)} is a channel system where any letter may be lost
at any moment from any channel.

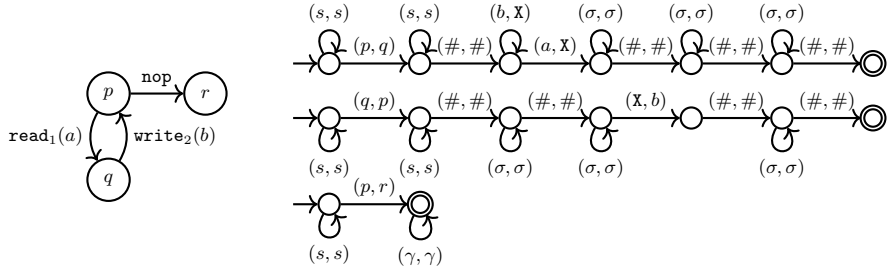
\begin{figure}
  \centering
  \begin{tikzpicture}[
  node distance=2.5cm, auto, thick, initial text={},
  scale=0.8, transform shape
]
  \tikzstyle{lstate} = [state, minimum size=20pt]

  %% LCS
  \node[lstate] (p) {$p$};
  \node[lstate, below=20pt of p] (q) {$q$};
  \node[lstate, right=25pt of p] (r) {$r$};

  \path[->]
  (p) edge[bend right=30] node[swap] {$\lcsread{a}{1}$} (q)
  (q) edge[bend right=30] node[swap] {$\lcswrite{b}{2}$} (p)
  (p) edge[]              node[]     {$\lcsnop$} (r)
  ;

  \tikzstyle{tstate} = [state, minimum size=10pt]

  %% Transducer for read
  \node[tstate, right=1.5cm of r, initial, yshift=0.5cm] (x0) {};
  \node[tstate, right=32pt of x0] (x1) {};
  \node[tstate, right=32pt of x1] (x2) {};
  \node[tstate, right=32pt of x2] (x3) {};
  \node[tstate, right=32pt of x3] (x4) {};
  \node[tstate, right=32pt of x4] (x5) {};
  \node[tstate, right=32pt of x5, accepting] (x6) {};

  \path[->]
  (x0) edge node {$(p, q)$} (x1)
  (x1) edge node {$(\#, \#)$} (x2)
  (x2) edge node {$(a, \texttt{X})$} (x3)
  (x3) edge node {$(\#, \#)$} (x4)
  (x4) edge node {$(\#, \#)$} (x5)
  (x5) edge node {$(\#, \#)$} (x6)

  (x0) edge[loop, out=120, in=60, looseness=8] node {$(s, s)$} (x0)
  (x1) edge[loop, out=120, in=60, looseness=8] node {$(s, s)$} (x1)
  (x2) edge[loop, out=120, in=60, looseness=8] node {$(b, \texttt{X})$} (x2)
  (x3) edge[loop, out=120, in=60, looseness=8] node {$(\sigma, \sigma)$} (x3)
  (x4) edge[loop, out=120, in=60, looseness=8] node {$(\sigma, \sigma)$} (x4)
  (x5) edge[loop, out=120, in=60, looseness=8] node {$(\sigma, \sigma)$} (x5)
  ;

  %% Transducer for write
  \node[tstate, below=15pt of x0, initial] (y0) {};
  \node[tstate, right=32pt of y0] (y1) {};
  \node[tstate, right=32pt of y1] (y2) {};
  \node[tstate, right=32pt of y2] (y3) {};
  \node[tstate, right=32pt of y3] (y4) {};
  \node[tstate, right=32pt of y4] (y5) {};
  \node[tstate, right=32pt of y5, accepting] (y6) {};

  \path[->]
  (y0) edge node {$(q, p)$} (y1)
  (y1) edge node {$(\#, \#)$} (y2)
  (y2) edge node {$(\#, \#)$} (y3)
  (y3) edge node {$(\texttt{X}, b)$} (y4)
  (y4) edge node {$(\#, \#)$} (y5)
  (y5) edge node {$(\#, \#)$} (y6)

  (y0) edge[loop, out=-120, in=-60, looseness=8] node[swap] {$(s, s)$} (y0)
  (y1) edge[loop, out=-120, in=-60, looseness=8] node[swap] {$(s, s)$} (y1)
  (y2) edge[loop, out=-120, in=-60, looseness=8] node[swap] {$(\sigma, \sigma)$} (y2)
  (y3) edge[loop, out=-120, in=-60, looseness=8] node[swap] {$(\sigma, \sigma)$} (y3)
  (y5) edge[loop, out=-120, in=-60, looseness=8] node[swap] {$(\sigma, \sigma)$} (y5)
  ;  

  %% Transducer for nop
  \node[tstate, below=30pt of y0, initial] (z0) {};
  \node[tstate, right=32pt of z0, accepting] (z1) {};

  \path[->]
  (z0) edge node {$(p, r)$} (z1)
  (z0) edge[loop, out=-120, in=-60, looseness=8] node[swap] {$(s, s)$} (z0)
  (z1) edge[loop, out=-120, in=-60, looseness=8] node[swap] {$(\gamma, \gamma)$} (z1)
  ;
\end{tikzpicture}
  \caption{\emph{Left}: Example of a lossy channel system with $\Gamma
    = \{a, b\}$. \emph{Right}: Transducers encoding respectively the
    transitions $\lcsread{a}{1}$, $\lcswrite{b}{2}$ and $\lcsnop$
    (under the semi-lossy semantics), where $s \in \{p, q, r\}$,
    $\sigma \in \{a, b\}$ and $\gamma$ stands for any
    letter.}\label{fig:lcs}
\end{figure}

For example, consider the channel system depicted on the left of
\Cref{fig:lcs}. Consider its configuration $C \defeq ([p, q], ab, \ew,
b)$ with two processes in states $p$ and $q$, and with three channels
currently containing $ab$, $\ew$ and $b$. We have, \eg,
\[
C
\lcstrans{\lcsread{a}{1}}  ([q, q], b, \ew, b)
\lcstrans{\lcswrite{b}{2}} ([q, p], b, b, b)
\lcstrans{\lcsnop} ([q, r], b, b, b).
\]

We encode configurations as words over $\Sigma \defeq P \cup \Gamma
\cup \{\lcsdelim\}$, \eg\ the encoding of $C$ is $pq \lcsdelim ab
\lcsdelim \lcsdelim a \lcsdelim$. Moreover, we encode operations as
transducers with an extra letter $\lcsdelete$. Let us explain this
through an example.

The transducers for the three transitions of our example are depicted
on the right of \Cref{fig:lcs}. As a preprocessing step for operation
$\lcswrite{\sigma}{i}$, we pad channels of a configuration on the
right with $\lcsdelete^*$. This way, the first occurrence of
$\lcsdelete$ can be replaced by the letter to insert. This padding can
then be removed in a postprocessing step. Furthermore, we model the
``lossiness'' of the channels through $\lcsread{\sigma}{i}$: when
reading an $a$, all letters in front of the first $a$ may be
lost. Transducers for $\lcswrite{\sigma}{i}$ and $\lcsnop$ are
nonlossy. For safety verification, the standard ``pure lossy''
semantics is equivalent to our ``semi-lossy'' semantics.

Let us formalize it. Given configurations $C = (X, w_1, \ldots, w_k)$
and $C' = (X', \allowbreak w_1', \allowbreak \ldots, \allowbreak
w_k')$, we write $C \preceq C'$ if $X = X'$ and $w_i \sqsubseteq w_i'$
for all $i \in [1..k]$, where ${\sqsubseteq}$ denotes the subword
order. For example, $([p, q], ab, \ew, b) \preceq ([p, q], aab, a,
baa)$. By Higman's lemma, the subword order ${\sqsubseteq}$ is a
well-quasi-order. Moreover, equality over a finite set is a
well-quasi-order. As the order ${\preceq}$ over configurations is a
product of these orders, it is also a well-quasi-order. Moreover:

\begin{restatable}{proposition}{propLossy}\label{prop:lossy}
  We have $C \lcstrans{*}_\text{pure-lossy} C''$ for some $C'' \succeq
  C'$ iff $C \lcstrans{*}_\text{semi-lossy} C''$ for some $C'' \succeq
  C'$. Furthermore, ${\lcstrans{}_\text{semi-lossy}}$ is monotone.
\end{restatable}

Note that $\cEnc{\upc{C}}$ is weakly acyclic. Indeed, for $C = ([p_1,
  \ldots, p_m], w_1, \ldots, w_k)$, the language $\cEnc{\upc{C}}$ can
be expressed by a weakly acyclic expression:
\[
p_1 \cdots p_m\ \lcsdelim
\prod_{i \in [1..k]}
\bigg(\prod_{a \in w_i}
(\Gamma \setminus \{a\})^* a\bigg) \Gamma^*\ \lcsdelim.
\]

\subsection{Petri nets}

A Petri net is a tuple $(P, T, F)$ where $P$ is a finite set of
\emph{places}, $T$ is a finite set of \emph{transitions} disjoint from
$P$, and $F \colon (P \times T) \cup (T \cup P) \to \N$. A
\emph{marking} is an assignment $\vec{m} \colon P \to \N$ of tokens to
places. A transition $t$ is \emph{enabled} in $\vec{m}$ if $\vec{m}(p)
\geq F(p, t)$ for every place $p$. If $t$ is enabled, then it can be
\emph{fired}, which leads to the marking $\vec{m}'$ obtained by
removing $F(p, t)$ tokens and adding $F(t, p)$ tokens to each place
$p$, \ie, $\vec{m}'(p) = \vec{m}(p) - F(p, t) + F(t, p)$.

\begin{figure}[b]
  \centering
  \begin{tikzpicture}[
  node distance=1.5cm, auto, thick, initial text={},
  scale=0.875, transform shape
]
  %% Petri net
  \node[place, label=below:$p$, tokens=3] (p) {};
  \node[transition, right of=p, label=below:$t$] (t) {};
  \node[place, right of=t, label=below:$q$, tokens=1] (q) {};

  \path[->]
  (p) edge[bend left=20] node {$2$} (t)
  (t) edge[bend left=20] node {$1$} (p)
  (q) edge[bend left=20] node {$1$} (t)
  (t) edge[bend left=20] node {$3$} (q)
  ;

  \tikzstyle{tstate} = [state, minimum size=10pt]

  %% Transducer
  \node[tstate, right=1cm of q, initial] (x0) {};
  \node[tstate, right=25pt of x0] (x1) {};
  \node[tstate, right=25pt of x1] (x2) {};
  \node[tstate, right=25pt of x2] (x3) {};
  \node[tstate, right=25pt of x3] (x4) {};
  \node[tstate, right=25pt of x4] (x5) {};
  \node[tstate, right=25pt of x5] (x6) {};
  \node[tstate, right=25pt of x6, accepting] (x7) {};

  \path[->]
  (x0) edge node {$(\pntok, \pntok)$} (x1)
  (x1) edge node {$(\pntok, \pnpad)$} (x2)
  (x2) edge node {$(\pndelim, \pndelim)$} (x3)
  (x3) edge node {$(\pntok, \pntok)$} (x4)
  (x4) edge node {$(\pnpad, \pntok)$} (x5)
  (x5) edge node {$(\pnpad, \pntok)$} (x6)
  (x6) edge node {$(\pndelim, \pndelim)$} (x7)

  (x1) edge[loop, out=-120, in=-60, looseness=8] node[swap] {$(\pntok, \pntok)$} (x1)
  (x4) edge[loop, out=-120, in=-60, looseness=8] node[swap] {$(\pntok, \pntok)$} (x4)
  ;
\end{tikzpicture}
  \caption{\emph{Left}: Example of a Petri net where $P = \{p, q\}$,
    $T = \{t\}$, $F(p, t) = 2$, $F(t, p) = 1 = F(q, t)$ and $F(t, q) =
    3$. \emph{Right}: Transducer encoding transition
    $t$.}\label{fig:pn}
\end{figure}
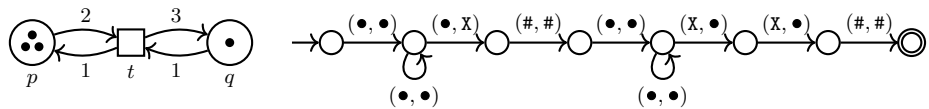

Consider, \eg, the Petri net depicted on the left of
\Cref{fig:pn}. From the marking $[p \mapsto 3, q \mapsto 1]$, written
as $\vec{m} \defeq (3, 1)$, we have, \eg, $\vec{m} \trans{t} (2, 3)
\trans{t} (1, 5)$.

We encode markings over $\Sigma \defeq \{\pntok, \pndelim\}$, \eg\ the
encoding of $\vec{m}$ is $\pntok \pntok \pntok \pndelim \pntok
\pndelim$. Moreover, we encode operations as transducers with an extra
letter $\pnpad$. Let us explain this through an example. The
transducer for the transition of our example is depicted on the right
of \Cref{fig:pn}. As a preprocessing step, we pad places of a
configuration on the right with $\pnpad^*$. This symbol can be used to
remove and add tokens. The padding is then removed in a postprocessing
step.

We write $\vec{m} \leq \vec{m}'$ if $\vec{m}(p) \leq \vec{m}'(p)$ for
all $p$. By Dickson's lemma, ${\leq}$ is a well-quasi-order. Note that
$\cEnc{\upc{\vec{m}}}$ is weakly acyclic. Indeed, for $\vec{m} = (n_1,
\ldots, n_k)$, it is the language of this weakly acyclic expression:
$\pntok^{n_1} \pntok^* \pndelim \cdots \pntok^{n_k} \pntok^*
\pndelim$.

\subsection{Broadcast protocols}
\label{subsec:broadcast}

A \emph{broadcast protocol} is a finite directed graph whose nodes $P$
are called \emph{states}, and where each arc (\emph{transition}) is
labeled by $\bpl{a}$, $\bpb{b}$, $\bpq{b}$, $\bpbb{c}$ or $\bpqq{c}$
where $a \in \Gamma$, $b \in \Gamma'$ and $c \in \Gamma''$ (three
disjoint alphabets). There are $m$ processes starting in some states,
and evolving by local, rendez-vous or broadcast communication:
\begin{itemize}
\item A local update $a$ only changes the state of a single process;

\item A rendez-vous occurs when two processes respectively take a
  $\bpb{b}$-transition and a $\bpq{b}$-transition;

\item A broadcast occurs when a process takes a $\bpbb{c}$-transition;
  when this happens, any other process that can take a
  $\bpqq{c}$-transition takes it.
\end{itemize}
For example, consider the broadcast protocol depicted on the left of
\Cref{fig:bp}. Let $C \defeq [p, p, p, p]$ be the configuration with
four processes in state $p$. We have, \eg,
$
C
\trans{b}
[q, r, p, p]
\trans{b}
[q, r, q, r]
\trans{c}
[p, q, p, r]
\trans{c}
[p, p, p, q]$.
We encode configurations as words over $\Sigma \defeq P$, \eg\ $[p, r,
  p, q]$ becomes $prpq$. The transducers for our example are depicted
on the middle and right of \Cref{fig:bp}.

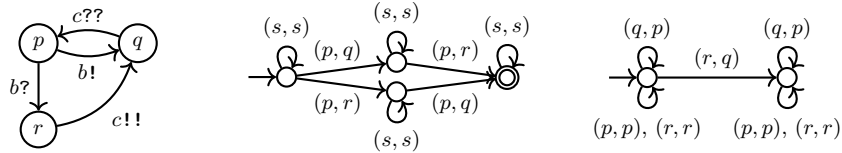
\begin{figure}
  \centering
  \vspace{-7pt}
  \begin{tikzpicture}[
  node distance=2.5cm, auto, thick, initial text={},
  scale=0.9, transform shape
]
  \tikzstyle{bstate} = [state, minimum size=15pt]

  %% Broadcast protocol
  \node[bstate] (p) {$p$};
  \node[bstate, below=20pt of p] (r) {$r$};
  \node[bstate, right=25pt of p] (q) {$q$};

  \path[->]
  (p) edge[bend right=20] node[swap] {$\bpb{b}$}  (q)
  (q) edge[bend right=20] node[swap] {$\bpqq{c}$} (p)
  (p) edge[]              node[swap] {$\bpq{b}$}  (r)
  (r) edge[bend right=30] node[swap] {$\bpbb{c}$} (q)
  ;

  \tikzstyle{tstate} = [state, minimum size=8pt, inner sep=1pt]

  %% Transducer for rendez-vous
  \node[tstate, right=1.75cm of q, initial, yshift=-0.5cm] (x0) {};
  \node[tstate, above right=0pt and 40pt of x0] (x0a) {};
  \node[tstate, below right=0pt and 40pt of x0] (x0b) {};
  \node[tstate, below right=0pt and 40pt of x0a, accepting] (x1) {};

  \path[->]
  (x0)  edge node[xshift=12pt]        {$(p, q)$} (x0a)
  (x0a) edge node[xshift=-12pt]       {$(p, r)$} (x1)
  (x0)  edge node[xshift=12pt,  swap] {$(p, r)$} (x0b)
  (x0b) edge node[xshift=-12pt, swap] {$(p, q)$} (x1)

  (x0)  edge[loop, out=120,  in=60,  looseness=8] node[]     {$(s, s)$} (x0)
  (x0a) edge[loop, out=120,  in=60,  looseness=8] node[]     {$(s, s)$} (x0a)
  (x0b) edge[loop, out=-120, in=-60, looseness=8] node[swap] {$(s, s)$} (x0b)
  (x1)  edge[loop, out=120,  in=60,  looseness=8] node[]     {$(s, s)$} (x1)
  ;

  %% Transducer for broadcast
  \node[tstate, right=1.75cm of x1, initial] (y0) {};
  \node[tstate, right=50pt of y0] (y1) {};

  \path[->]
  (y0) edge node[] {$(r, q)$} (y1)

  (y0) edge[loop, out=120, in=60, looseness=8] node[] {$(q, p)$} (y0)
  (y1) edge[loop, out=120, in=60, looseness=8] node[] {$(q, p)$} (y1)

  (y0) edge[loop, out=-120, in=-60, looseness=8] node[swap] {
    \begin{tabular}{c}
      $(p, p)$, $(r, r)$
    \end{tabular}
  } (y0)

  (y1) edge[loop, out=-120, in=-60, looseness=8] node[swap] {
    \begin{tabular}{c}
      $(p, p)$, $(r, r)$
    \end{tabular}
  } (y1)
  ;
\end{tikzpicture}
  \vspace{-5pt}
  \caption{\emph{Left}: Example of a broadcast protocol with $P = \{p,
    q, r\}$, $\Gamma = \emptyset$, $\Gamma' = \{b\}$ and $\Gamma'' =
    \{c\}$. \emph{Middle and right}: Transducers encoding the
    transitions, where $s \in \{p, q, r\}$.}\label{fig:bp}
\end{figure}

\section{Experimental results}
\label{sec:exp}
\newcommand{\cpp}{C\raisebox{1.5pt}{\scriptsize++}}

We developed a prototype \cpp{} library for weakly acyclic diagrams,
which we refer to as \wadl{}\footnote{Our artifact is available at \url{https://doi.org/10.5281/zenodo.14624775}.}. It implements all operations described
in \Cref{sec:oper}, and backwards reachability for lossy channel
systems, Petri nets and broadcast protocols, as in
\Cref{sec:model:check}. To benchmark the performance, we
use four different sets:
\begin{itemize}
\item \emph{Lossy channel systems.} We collected 13 instances used by
  the tools \bml{}~\cite{GeffroyLS17} and
  \mcscm{}~\cite{HeussnerGS09,HeussnerGS12}. They arise from a mutual
  exclusion algorithm and from protocols (communication, business,
  etc.)  The instance \texttt{ring2} comes from a parameterized
  family. Thus, we generalized it to \texttt{ring3}, \texttt{ring4},
  etc.

\item \emph{Petri nets.} We collected 173 instances used in many
  papers on Petri net safety verification (\eg\ see~\cite{ELMMN14}). They
  model protocols (mutual exclusion, communication, etc.); concurrent
  C and Erlang programs; provenance analysis of messages of a medical
  messaging system and a bug-tracking system.

\item \emph{Broadcast protocols.} We collected 38 broadcast protocol
  instances from the benchmark set of \dodo{}, a tool for regular
  model checking that uses the encoding of Section
  \ref{subsec:broadcast} for broadcast protocols
  \cite{CzernerEKW24,Welzel-Mohr24}. The instances correspond to
  safety properties of several cache-coherence protocols.
  
\item \emph{Regular model checking problems.} Backwards reachability
  can also be applied to systems which are not well structured, as
  long as all the intermediate sets of configurations remain weakly
  acyclic. The algorithm may not terminate, but if it does then it
  yields the correct answer. We collected 35 further instances from
  \dodo{} modeling parameterized mutual exclusion algorithms, variants
  of the dining philosophers, algorithms for termination detection and
  leader election, and token passing algorithms.
\end{itemize}
 
We carried experiments for the four above categories with the goal of
testing the potential and versatility of WADs. We used a machine with
an 8-Core $\text{11}^\text{th}$ Gen Intel® Core™ i7-1165G7 CPU @
2.80GHz running on Kubuntu 24.04 with 31.1GiB of memory. The time was
determined with the sum of the \texttt{user} and \texttt{sys} times
given by Linux tool \texttt{time}. We used a timeout of 10 minutes per
instance, except for Petri nets where we used a timeout of 20
minutes\footnote{During the reviewing process, a reviewer requested a
longer timeout for this case.}.

\paragraph{Lossy channel systems.}

We made a few optimizations compared to \Cref{sec:model:check}. First,
states are represented in binary to avoid using $|P|$ letters. Second,
we tweaked the transducers so that they yield the pre-compatibility of
\Cref{ssec:wa:rel}. The transducers of \Cref{fig:lcs} do not
necessarily satisfy it as the initial states use letters $(x, y)$ and
$(x', y')$ with $x = x'$. This can be fixed, \eg\ make the initial
state of the top transducer loop on $(q, q)$ and $(r, r)$, but not
$(p, p)$. This slightly changes the semantics: only the first process
in $p$ can execute. Yet, as processes in a common state are
indistinguishable, only the state count matters and so the change does
not alter safety verification. Moreover, in our benchmarks, processes
evolve in disjoint subgraphs and hence can never even be in a common
state.

We compared with \bml{}~\cite{GeffroyLS17} which runs the backward
algorithm with pruning (modes: \textsc{csre}, \textsc{mof},
\textsc{si}). Pruning can be disabled to yield the basic algorithm. We
also compared to \mcscm{}~\cite{HeussnerGS09,HeussnerGS12} which (in
its default mode) uses an adaptive extrapolated backwards
computation. The results are depicted on the left of
\Cref{fig:results:1}. \wadl{} solves the most instances and does so
generally faster, although the competition is close. We further tested
the same tools on \texttt{ring3}, \texttt{ring4}, etc. Our tool
terminates on large $n$ (e.g.\ \texttt{ring99} in 2m38s and
\texttt{ring187} is the largest solved), but the other tools did not
scale, as depicted in \Cref{fig:results:2}.

\begin{figure}
  \centering
  \begin{minipage}[t]{0.47\textwidth}
    \plotcanvas{time $t$ in seconds}{number of instances \\ decided in $\leq t$ seconds}{%
        \toolplotwithtotalcount{wadl}{\markwadl}{
          (0.1, 5)(0.25, 8)(0.5, 9)(1.0, 10)(5.0, 11)(10.0, 11)(15.0, 11)(50.0, 12)(150.0, 12)(600.0, 12)
        }{2800.0}{12}{12/13};
        \toolplotwithtotalcount{mcscm}{\markmcscm}{
          (0.1, 1)(0.25, 4)(0.5, 6)(1.0, 6)(5.0, 7)(10.0, 7)(15.0, 7)(50.0, 7)(150.0, 8)(600.0, 8)
        }{600.0}{8}{8/13};      
        \toolplotwithtotalcount{bmlMof}{\markbmlMof}{
          (0.1, 2)(0.25, 5)(0.5, 7)(1.0, 10)(5.0, 11)(10.0, 11)(15.0, 11)(50.0, 12)(150.0, 12)(600.0, 12)
        }{1300.0}{12}{12,\phantom{/13}};
        \toolplotwithtotalcount{bml}{\markbml}{
          (0.1, 3)(0.25, 4)(0.5, 6)(1.0, 7)(5.0, 10)(10.0, 10)(15.0, 10)(50.0, 10)(150.0, 10)(600.0, 11)
        }{1300.0}{10.5}{11/13};
        \toolplotwithtotalcount{bmlCsre}{\markbmlCsre}{
          (0.1, 4)(0.25, 7)(0.5, 7)(1.0, 8)(5.0, 10)(10.0, 10)(15.0, 10)(50.0, 11)(150.0, 11)(600.0, 11)
        }{600.0}{10.5}{11,\phantom{/13}};      
        \toolplotwithtotalcount{bmlSi}{\markbmlSi}{
          (0.1, 2)(0.25, 3)(0.5, 5)(1.0, 6)(5.0, 7)(10.0, 10)(15.0, 10)(50.0, 12)(150.0, 12)(600.0, 12)
        }{600.0}{12}{12,\phantom{/13}};
      }
  \end{minipage}\quad%
  \begin{minipage}[t]{0.47\textwidth}
    \plotcanvasext{time $t$ in seconds}{number of instances \\ decided in $\leq t$ seconds}{%
        \toolplotwithtotalcount{mist}{\markmist}{
          (0.1, 13)(0.25, 23)(0.5, 26)(1.0, 29)(5.0, 41)(10.0, 42)(15.0, 47)(50.0, 49)(150.0, 54)(600.0, 60)(1200.0, 60)
        }{1200.0}{60}{60/173};
        \toolplotwithtotalcount{besttool}{\markbesttool}{
          (0.1, 18)(0.25, 25)(0.5, 28)(1.0, 31)(5.0, 45)(10.0, 45)(15.0, 52)(50.0, 55)(150.0, 59)(600.0, 85)(1200.0, 91)
        }{1200.0}{94.5}{91/173};
        \toolplotwithtotalcount{wadl}{\markwadl}{
          (0.1, 17)(0.25, 20)(0.5, 25)(1.0, 28)(5.0, 40)(10.0, 41)(15.0, 44)(50.0, 55)(150.0, 56)(600.0, 80)(1200.0, 89)
        }{1200.0}{88}{89/173};
      }
  \end{minipage}
  
  \generatelegend
  \medskip
  
  \caption{Cumulative number of instances decided over time (semi-log
    scale) for lossy channel systems excluding $\{\texttt{ring3},
    \texttt{ring4}, \ldots\}$ (left), and Petri nets (right). 
    \besttool\ refers to the best outcome on instances solved by at
    least one of the two tools.}\label{fig:results:1}
\end{figure}
\begin{table}
  \centering
  \resizebox{0.65\columnwidth}{!}{%
  \begin{tabular}{lrrrrrr}
    \toprule
    & \bml & \bmlCsre & \bmlMof & \bmlSi & \mcscm & \wadl \\
    \midrule
    \texttt{ring2} & 0.63 & 1.02 & 0.89 & 0.37   & \TO    & \best0.04 \\
    \texttt{ring3} & \TO  & \TO  & \TO  & 9.13   & 416.48 & \best0.10 \\
    \texttt{ring4} & \TO  & \TO  & \TO  & 307.98 & \SO    & \best0.19 \\
    \texttt{ring5} & \TO  & \TO  & \TO  & \TO    & \TO    & \best0.33 \\
    \bottomrule
  \end{tabular}}
  \medskip
  
  \caption{Time in seconds to solve \texttt{ring}$n$ (\TO~=~timeout,
    \SO~=~stack overflow).}\label{fig:results:2}
\end{table}

\paragraph{Petri nets.} We compared with \mist{}~\cite{DRB04,GMDKRV07}, an
established efficient tool that, in particular, offers an
implementation of the backward algorithm without pruning, as in our
case, but using interval sharing trees as the data structure. The
results are depicted on the right of \Cref{fig:results:1}. \mist{} is
competitive, but we solved more instances in total. There are
31~instances that we solved that \mist{} did not; 2~instances for the
other way around; and 58~instances solved by both tools.

\paragraph{Broadcast protocols.}

We solved 32/38 instances, with 25 solved within 100ms. We do not
compare with \dodo{} as its goal is not to check the property, but
rather to compute small invariants explaining why it holds. So,
although we would look good in comparison, it would be meaningless
(the runtimes of \dodo{} reported in~\cite{Welzel-Mohr24} are at least
ten times higher than ours in every instance).

\paragraph{Regular model checking.}

While there is \emph{a priori} no reason why the instances should be
checkable with WADs, we solved 22/35 instances (all within 10ms). We
further timed out on 6~instances, and could not verify
7~instances. For the latter, a cycle was contracted in
Line~\ref{ln:S:marked} of \Cref{alg:pre:gen}, and hence we cannot be
certain that the result is weakly acyclic.

The fact we solved many instances suggests that weakly acyclic
languages may be common in regular model checking. In fact, popular
examples used to introduce the technique, like the token-passing
protocol, FIFO channel, and stack examples of the
surveys~\cite{AJNS04,Abd12,AST18,Abd21}, or the alternating bit
protocol of~\cite{AbdullaJ93} are modeled with weakly acyclic
languages.

\section{Conclusion}
\label{sec:conc}
We introduced weakly acyclic diagrams as a general data structure that
extends binary decision diagrams to (possibly) infinite languages,
while maintaining their algorithmic advantages. Many instances in
regular model checking fall within the class of weakly acyclic
languages. Moreover, weakly acyclic diagrams allow for the
manipulation and verification of various well-structured transition
systems such as lossy channel systems, Petri nets and broadcast
protocols.

As future work, it would be interesting to see whether the good
algorithmic properties can be retained under a slight extension of
nondeterminism, \eg\ for languages like $(a + b)^*b$ which are simple
but not weakly acyclic per se.

%%\begin{credits}
%%  \subsubsection{\ackname} Acknowledgments here.
%%\end{credits}

%% References
\bibliographystyle{splncs04}
\bibliography{references}

%% Appendix
\clearpage
\appendix
\section{Missing proofs of \Cref{sec:wa}}

\propAcycPow*

\begin{proof}
  Let $\A = (Q, \Sigma, \delta, q_0, F)$ be a weakly acyclic
  NFA. Recall that the states of the DFA obtained by the powerset
  construction are subsets of states of $\A$, and we have a transition
  $Q_1 \trans{a} Q_2$ in the DFA if{}f $\delta(Q_1, a) = Q_2$, where
  $\delta(R, a) \defeq \bigcup_{q \in R} \delta(q, a)$. So, by
  definition of weakly acyclic DFAs, it suffices to show that the
  following binary relation $\preceq$ over $2^Q$ is a partial order:
  \[
  Q_1 \preceq Q_2 \iff \delta(Q_1, w) = Q_2 \text{ for some word $w$}.
  \]
  The relation is reflexive because $\delta(Q_1, \ew) = Q_1$, and
  transitive since $\delta(Q_1, w_1) = Q_2$ and $\delta(Q_2, w_2) =
  Q_3$ implies $\delta(Q_1, w_1w_2) = Q_3$. It remains to show that
  $\preceq$ is antisymmetric, \ie, that $\delta(Q_1, w_1) = Q_2$ and
  $\delta(Q_2, w_2) = Q_1$ implies $Q_1 = Q_2$.

  Assume $\delta(Q_1, w_1) = Q_2$ and $\delta(Q_2, w_2) = Q_1$. We say
  that a state $q \in Q_1$ is \emph{cyclic} if there exists some $n
  \geq 1$ such that $q \in \delta(q, (w_1w_2)^n)$. We prove that every
  state of $Q_1$ is cyclic, which establishes $Q_1 = Q_2$.

  Assume $Q_1$ has an acyclic state $q$. We can pick $q$ minimal
  w.r.t.\ ${\preceq}$. Since $\delta(Q_1, w_1w_2) = Q_1$ by
  assumption, there is some $q' \in Q_1$ such that $q \in \delta(q',
  w_1w_2)$, and so $q \succeq q'$. As $q$ is acyclic we have $q' \neq
  q$, and so $q \succ q'$. By minimality of $q$, the state $q'$ is
  cyclic. As $\A$ is weakly acyclic, $\delta(q', a) = \{q'\}$ for all
  $a \in \alp{w_1w_2}$, and so, in particular, $\delta(q', w_1w_2) =
  \{q'\}$. This contradicts $q \neq q'$. \qed
\end{proof}

\propAcycMinDfa*

\begin{proof}
  Let $\A = (Q, \Sigma, \delta, q_0, F)$ be a weakly acyclic
  DFA. Given states $p, q \in Q$, we write $p \preceq q$ if $p
  \trans{w} q$ for some word $w \in \Sigma^*$. Note that $\preceq$ is
  a partial order.

  Let $\A'$ be the minimal DFA for $L \defeq \lang{\A}$. Towards a
  contradiction, suppose that $\A'$ is not weakly acyclic. It contains
  a simple cycle of length at least two, \ie, there are residuals
  $L^{x} \neq L^{y}$ and $w_1, w_2 \in \Sigma^+$ with $L^{xw_1} = L^y$
  and $L^{yw_2} = L^x$.
  
  Let $w \defeq w_1 w_2$ and $p_n \defeq \delta(q_0, xw^n)$. The set $\{p_n : n \in \N\}$ is finite, so it has at least one maximal state w.r.t.\ ${\preceq}$. Let $n \in \N$ be such that $p_n$ is maximal. As $p_n \trans{w} p_{n+1}$, we have $p_n \preceq p_{n+1}$, and by maximality, $p_n = p_{n+1}$.
  
  Now, define $q \defeq \delta(p_n, w_1)$. Note that $\delta(q, w_2) = \delta(p_n, w) = p_{n+1} = p_n$, so $p_n$ and $q$ lie on a cycle in $\A$. Thus, for a contradiction to the weak acyclicity of $\A$, it remains to show that $p_n \neq q$. We have $L^{xw} = L^{x}$, so we get inductively
  $
  \lang{p_n} = L^{xw^n} = L^x
  $.
  On the other hand, this implies
  $
  \lang{q} = \lang{p_n}^{w_1} = L^{xw_1} = L^y \neq L^x
  $.
  Thus, $L(p_n) \neq L(q)$, which yields the contradiction $p_n \neq q$. \qed
\end{proof}

\lemResiduals*

\begin{proof}
  $\Rightarrow)$ Let $L \subseteq \Sigma^*$ be a weakly acyclic
  language. Let $u \in \Sigma^*$ and $v \in \Sigma^+$ be such that
  $L^u = L^{uv}$. Let $\A = (Q, \Sigma, \delta, q_0, F)$ be the
  minimal weakly acyclic DFA that accepts $L$. Let $q, r \in Q$ be the
  states such that $q_0 \trans{u} q \trans{v} r$. Since $L^u =
  L^{uv}$, $|v| > 0$, and $\A$ is weakly acyclic, we have $q = r$ and
  $q \trans{a} q$ for every $a \in \alp{v}$. Thus, $L^u = L^{uw}$ for
  every $w \in \alp{v}^*$.

  $\Leftarrow$) Let $\A$ be the minimal DFA that accepts $L$. We must
  show that $\A$ is weakly acyclic. For the sake of contradiction,
  suppose that $q_0 \trans{u} q \trans{a} r \trans{x} q$ holds for
  some states $q \neq r$, $u \in \Sigma^*$, $a \in \Sigma$ and $x \in
  \Sigma^+$. We have $L^u = L^{uv}$ where $v \defeq ax$. By
  assumption, this implies that $L^u = L^{uw}$ for every $w \in
  \alp{v}^*$. In particular, $L^u = L^{ua}$. By minimality of $\A$,
  each state accepts a distinct residual. Hence, we must have $q_0
  \trans{u} q \trans{a} q$, which contradicts $q \neq r$. \qed
\end{proof}

\begin{proposition}\label{prop:acyc:expr:dfa}
  For every weakly acyclic expression $r$, there is an acyclic DFA
  that accepts $\lang{r}$.
\end{proposition}

\begin{proof}
  We proceed by structural induction on expression $r$. The claim is
  obvious for $r = \emptyset$ and $r = \Gamma^*$. Assume $r =
  \Lambda^* ar$ for some $\Lambda \subseteq \Sigma$ and $a \notin
  \Lambda$. By induction, there exists a weakly acyclic DFA $\A = (Q,
  \Sigma, \delta, q_0, F)$ such that $\lang{A} = \lang{r}$. The
  following weakly acyclic DFA accepts $\Lambda^* a \lang{r}$:

  \begin{center}
    \begin{tikzpicture}[
        node distance=2cm, auto, thick, initial text={},
        scale=0.8, transform shape
      ]
      \node[state, initial]           (q0) {};
      \node[state, above right=5pt and 2cm of q0] (q1) {$q_0$};
      \node[state, below right=5pt and 2cm of q0] (q2) {};
      \node[right=4cm of q1]          (q3) {};
      
      \node[fit=(q1)(q3), draw, inner sep=5pt, label={below:\Large$\A$}] {};
      
      \path[->]
      (q0) edge[loop above] node {$\Lambda$} () 
      (q0) edge[] node {$a$} (q1)
      (q0) edge[] node[swap, xshift=15pt, yshift=-5pt] {$\Sigma \setminus (\Lambda \cup \{a\})$} (q2)
      (q2) edge[loop right] node {$\Sigma$} () 
      ;
    \end{tikzpicture}
  \end{center}

  Assume $r = r_1 + r_2$. By induction, there are weakly acyclic DFAs
  $\A_1 = (Q_1, \allowbreak \Sigma, \allowbreak \delta_1, q_{01},
  F_1)$ and $\A_2 = (Q_2, \Sigma, \delta_2, q_{02}, F_2)$ such that
  $\lang{A_i} = \lang{r_i}$ for both $i \in \{1, 2\}$. We claim that
  the usual pairing $\A = [\A_1, \A_2]$ for $\lang{A_1} \cup
  \lang{A_2}$ is acyclic. To obtain a contradiction, suppose it is
  not. There are states $[p_1, p_2] \neq [q_1, q_2]$ and words $u, v
  \in \Sigma^+$ such that $[p_1, p_2] \trans{u} [q_1, q_2] \trans{v}
      [p_1, p_2]$.
  
  Let $i \in \{1, 2\}$ be such that $p_i \neq q_i$. By definition of
  pairing, we have $\delta(p_i, u) = q_i$ and $\delta(q_i, v) =
  p_i$. So, $p_i$ and $q_i$ are distinct states of $\A_i$ that belong
  to the same strongly connected component. This contradicts the weak
  acyclicity of $\A_i$. \qed
\end{proof}

\begin{proposition}\label{prop:acyc:nfa:expr}
  For every weakly acyclic NFA $\A$, there is a weakly acyclic
  expression for $\lang{A}$.
\end{proposition}

\begin{proof}
  Let $\A = (Q, \Sigma, \delta, q_0, F)$. Without its self-loops, $\A$
  is a directed acyclic graph, and hence has finitely many paths. So,
  $\lang{A}$ is a finite union of languages of the form $\Lambda_1^*
  a_1 \cdots \Lambda_n^* a_n \Gamma^*$ where $\Lambda_1, \ldots,
  \Lambda_n, \Gamma \subseteq \Sigma$ and each $a_i \notin
  \Lambda_i$. \qed
\end{proof}

\begin{corollary}
  A language is weakly acyclic iff it is described by a weakly
  acyclic expression.
\end{corollary}

\section{Missing proofs of \Cref{sec:ds}}

\propFundOne*

\begin{proof}
  The result clearly holds for $\emptyset$ and $\Sigma^*$. Further,
  assume that for every $L^a$ such that $L^a \neq L$ the language
  accepted from state $L^a$ is $L^a$. Since for each such letter $a$
  we have $\delta_{\M}(L, a) = L^a$, for each other letter $b$ we have
  $\delta_{\M}(L, b) = L$, and state $L$ is accepting iff $\ew \in L$,
  the language recognized from state $L$ is $L$. \qed
\end{proof}

\propFundTwo*

\begin{proof}
  By definition, states of the master automaton are distinct
  languages. By \Cref{prop:fund1}, distinct states of $\A_L$ accept
  distinct languages. Recall that a DFA is minimal iff its states
  accept different languages. So, $\A_L$ is minimal. \qed
\end{proof}

\propMake*

\begin{proof}
  For the sake of contradiction, let us consider a table obtained by
  successive calls to \make, and let $\make(s, b)$ be the first call
  that creates a state $q$ whose language is the same as another
  existing state $q' \neq q$. Note that $q'$ is unique, as otherwise
  there would already have been two states with the same language.

  Let $(s', b')$ be the entry of $q'$ in the table. As $\lang{q} =
  \lang{q'}$, we have $b = b'$. Moreover, we claim that the following
  properties hold for all $a \in \Sigma$:
  \begin{enumerate}
  \item[(i)] if $q^a \notin \{\self, q'\}$, then $(q')^a =
    q^a$;\label{itm:claim:1}

  \item[(ii)] if $q^a \in \{\self, q'\}$, then $(q')^a =
    \self$.\label{itm:claim:2}
  \end{enumerate}

  Before proving the claim, let us see why it allows us to prove the
  proposition. If $q' \notin s$, then, by~(i) and~(ii), we have $s =
  s'$, and so Line~\ref{ln:make:exists} of \make is executed, which is
  a contradiction as $q$ is a new state. Thus, we may assume that $q'
  \in s$. By~(i) and~(ii), we have $s[q' / \self] = s'$. Let $q''
  \defeq \max\{r \in s : r \neq \self\}$. If $q' = q''$, then
  Line~\ref{ln:make:equiv} of \make is executed, which yields a
  contradiction. So, we may assume that $q' < q''$. Let $a \in \Sigma$
  be such that $q^a = q''$. We have $\lang{q'}^a = \lang{q}^a =
  \lang{q^a} = \lang{q''}$. Thus, by unicity, we have $(q')^a =
  q''$. This means that $q'$ has a successor $q''$ such that $q' <
  q''$, which is impossible. Therefore, we get a contradiction in all
  cases, which proves the proposition.

  It remains to prove~(i) and~(ii).

  \medskip\emph{Property~(i)}. Let $a \in \Sigma$ be such that $q^a
  \notin \{\self, q'\}$. Since $q'$ is the unique state of the table
  such that $\lang{q'} = \lang{q}$, we must have $\lang{q^a} \neq
  \lang{q}$. If $(q')^a = \self$, then $\lang{q^a} = \lang{q}^a =
  \lang{q'}^a = \lang{q'} = \lang{q}$, which is a contradiction. Thus,
  we may assume that $(q')^a \neq \self$.

  We have $\lang{q^a} = \lang{q}^a = \lang{q'}^a =
  \lang{(q')^a}$. Note that $q^a$ existed prior to the creation of
  $q$. Thus, we must have $q^a = (q')^a$ as otherwise there would be
  two distinct states accepting the same language.

  \medskip\emph{Property~(ii)}. Let $a \in \Sigma$ be such that $q^a
  \in \{\self, q'\}$. If $q^a = \self$, then $\lang{q'}^a = \lang{q}^a
  = \lang{q} = \lang{q'}$. If $q^a = q'$, then $\lang{q'}^a =
  \lang{q}^a = \lang{q^a} = \lang{q'}$. Therefore, in both cases, we
  have $\lang{q'}^a = \lang{q'}$. This implies $(q')^a = \self$, as
  otherwise $(q')^a$ would be a distinct state accepting the same
  language as $q'$. \qed
\end{proof}

\section{Missing proofs of \Cref{ssec:binop}}

Let us prove the claims of \Cref{prop:comp:inter:full} separately.

\begin{proposition}
  \Cref{alg:comp} is correct.
\end{proposition}

\begin{proof}
  Let $L \defeq \lang{q}$. The result $r$ of \Cref{alg:comp} should
  satisfy $\lang{r} = \overline{L}$. Let us show that this is the
  case.

  Note that $\ew \in \overline{L}$ iff $\ew \notin L$. Therefore,
  ``$\neg \mathrm{flag}(q)$'', on Line~\ref{ln:comp:return}, indicates
  correctly if the empty word belongs to $\overline{L}$.

  Let $a \in \Sigma$. \Cref{alg:comp} must assign to $s_a$ the state
  accepting $(\overline{L})^a$. For every $u \in \Sigma^*$, we have
  \begin{align*}
    u \in (\overline{L})^a
    &\iff au \in \overline{L}
    \iff au \notin L
    \iff u \notin L^a
    \iff u \in \overline{L^a}.    
  \end{align*}

  If $q^a = \self$, then $L^a = L$, and so by the above we have
  $(\overline{L})^a = \overline{L}$. Thus, the algorithm correctly
  assigns ``$s_a \leftarrow \self$''.

  Otherwise, by the above, assigning ``$s_a \leftarrow
  \comp{$q^a$}$'' is correct. \qed
\end{proof}

\begin{proposition}
  \Cref{alg:comp} with memoization works in time $\O(mn)$, where $m
  \defeq |\Sigma|$ and $n$ is the number of nodes reachable from $q$.
\end{proposition}

\begin{proof}
  Due to memoization, Line~\ref{ln:comp:for} is reached at most once
  per node $q'$ reachable from $q$. Moreover, the loop iterates over
  each letter $a \in \Sigma$. \qed
\end{proof}

\begin{proposition}
  \Cref{alg:inter} is correct.
\end{proposition}

\begin{proof}
  Let $K \defeq \lang{p}$ and $L \defeq \lang{q}$. The result $r$ of
  \Cref{alg:inter} should satisfy $\lang{r} = K \cap L$. Let us show
  that this is the case.

  Note that $\ew \in K \cap L$ iff $\ew \in K \land \ew \in
  L$. Therefore, ``$\mathrm{flag}(p) \land \mathrm{flag}(q)$'', on
  Line~\ref{ln:inter:return}, indicates correctly if the empty word
  belongs to $K \cap L$.

  Let $a \in \Sigma$. \Cref{alg:inter} must assign to $s_a$ the state
  accepting $(K \cap L)^a$. For every $u \in \Sigma^*$, we have
  $
  u \in (K \cap L)^a
  \iff au \in K \cap L
  \iff (au \in K) \land (au \in L)
  \iff u \in K^a \land u \in L^a.    
  $

  If $p^a = q^a = \self$, then $L^a = L$ and $K^a = K$, and so by the
  above we have $(K \cap L)^a = K \cap L$. Thus, the algorithm
  correctly assigns ``$s_a \leftarrow \self$''.

  Otherwise, by the above, assigning ``$s_a \leftarrow \inter{$p',
    q'$}$'' is correct. \qed
\end{proof}

\begin{proposition}
  \Cref{alg:inter} with memoization works in time $\O(m n_p n_q)$,
  where $m \defeq |\Sigma|$ and $n_x$ is the number of nodes reachable
  from node $x$.
\end{proposition}

\begin{proof}
  Due to memoization, Line~\ref{ln:inter:for} is reached at most once
  per pair $(x, y)$ where $x$ is a node reachable from $p$, and $y$ is
  a node reachable for $q$. Moreover, the loop iterates over each
  letter $a \in \Sigma$. \qed
\end{proof}

\section{Missing proofs of \Cref{ssec:rel}}

\propSRes*

\begin{proof}
  We consider $w = a \in \Sigma$. The general case follows by
  induction. We have
  \begin{alignat*}{3}
    &&& u \in \lang{S}^a \\
    &\iff\ && au \in \lang{S} \\
    &\iff\ &&
    au \in \bigcup_{(p, q) \in S} \Pre{p}{q} \\
    &\iff\ &&
    \exists b \in \Sigma, v \in \Sigma^{|u|}, (p, q) \in S
    : (au, bv) \in \lang{p} \land bv \in \lang{q} \\
    &\iff\ &&
    \exists b \in \Sigma, v \in \Sigma^{|u|}, (p, q) \in S
    : (u, v) \in \lang{p}^{(a, b)} \land v \in \lang{q}^b \\
    &\iff\ &&
    \exists b \in \Sigma, v \in \Sigma^{|u|}, (p, q) \in S,
    p' \in \delta(p, (a, b))
    : (u, v) \in \lang{p'} \land v \in \lang{\tsucc{$q, b$}} \\
    &\iff\ &&
    \exists b \in \Sigma, (p, q) \in S, p' \in \delta(p, (a, b))
    : u \in \Pre{p'}{\tsucc{$q, b$}} \\
    &\iff\ &&
    u \in \bigcup_{(p', q') \in S'} \Pre{p'}{q'} \\
    &\iff\ && u \in \lang{S'}. \tag*{\qed}
  \end{alignat*}
\end{proof}

\propSCycle*

\begin{proof}
  Let $v' \in \Sigma^*$ be such that $S' \trans{v'} S$. Let $L \defeq
  \lang{S_0}$. Let $u \in \Sigma^*$ be such that $S_0 \trans{u}
  S'$. Let $v \defeq v'a$. Since $S_0 \trans{u} S' \trans{v} S'$, by
  \Cref{prop:s:res}, we have $L^u = \lang{S'}$ and $L^{uv} =
  \lang{S'}$. So, $L^u = L^{uv}$. Since $L$ is weakly acyclic, by
  \Cref{lem:residuals}, we have $L^u = L^{uw}$ for all $w \in
  \alp{v}^*$. Thus, $L^u = L^{uv'}$. Since $S_0 \trans{uv'} S$, by
  \Cref{prop:s:res}, we have $L^{uv'} = \lang{S}$. Thus, $\lang{S} =
  L^{u v'} = L^{u} = \lang{S'}$. \qed
\end{proof}

\begin{proposition}
  \Cref{alg:pre:gen} terminates.
\end{proposition}

\begin{proof}
  For the sake of contradiction, suppose that the algorithm does not
  terminate on some input $S_0$. There exists an infinite branch $S_0
  \trans{a_0} S_1 \trans{a_1} S_2 \trans{a_2} \cdots$ generated by
  successive calls to Line~\ref{ln:S:unmarked}. Note that $S_j \neq
  S_k$ for all $j \neq k$, as otherwise Line~\ref{ln:S:marked} would
  be executed at some point.

  For every $i \in \N$, let $Q_i \defeq \{q : (p, q) \in S_i\}$. By
  definition on Line~\ref{ln:S:def}, for every $q' \in Q_{i+1}$, there
  exist $b \in \Sigma$ and $q \in Q_i$ such that $\lang{q'} =
  \lang{q}^b$. Thus, for every $q' \in Q_i$, there exists $q \in Q_0$
  such that $\lang{q'}$ is a residual of $\lang{q}$. Since $Q_0$ is
  finite and a regular language has finitely many residuals, the set
  $Q \defeq Q_0 \cup Q_1 \cup \cdots$ is finite. Let $P_i \defeq \{p :
  (p, q) \in S_i\}$. Since the transducer is finite, it must be the
  case that $P \defeq P_0 \cup P_1 \cup \cdots$ is finite.

  Since $S_i \subseteq P_i \times Q_i \subseteq P \times Q$ for all $i
  \in \N$, and since $P \times Q$ is finite, there must exist $j \neq k$
  such that $S_j = S_k$. This is a contradiction. \qed
\end{proof}

Let us prove the claims of \Cref{prop:pre:comp:full} separately.

\begin{proposition}\label{prop:pre:comp:correct}
  \Cref{alg:pre:comp} is correct.
\end{proposition}

\begin{proof}
  Let $K \defeq \lang{p}$ and $L \defeq \lang{q}$. The result $r$ of
  \Cref{alg:pre:comp} should satisfy $\lang{r} = \Pre{K}{L}$. Let us
  show that this is the case.

  Note that $\ew \in \Pre{K}{L}$ iff $(\ew, \ew) \in K \land \ew \in
  L$. Therefore, ``$\flag{p} \land \flag{q}$'', on
  Line~\ref{ln:pre:return}, indicates correctly if the empty word
  belongs to $\Pre{K}{L}$.
  
  Let $a \in \Sigma$. \Cref{alg:pre:comp} must assign to $s_a$ the
  state accepting $(\Pre{K}{L})^a$. For every $u \in \Sigma^*$, we
  have
  \begin{align*}
    u \in (\Pre{K}{L})^a
    &\iff
    au \in \Pre{K}{L} \\
    &\iff
    \exists b \in \Sigma \text{ and }
    v \in \Sigma^{|u|} : (au, bv) \in K \land bv \in L \\
    &\iff
    \exists b \in \Sigma \text{ and }
    v \in \Sigma^{|u|} : (u, v) \in K^{(a, b)} \land v \in L^b \\
    &\iff
    \exists b \in \Sigma : u \in \Pre{K^{(a, b)}}{L^b}.
  \end{align*}
  So, if there is no $b \in \Sigma$ such that $p^{(a, b)} \neq
  p_\emptyset$ and $q^{(a, b)} \neq q_\emptyset$, then $(\Pre{K}{L})^a
  = \emptyset$, and hence assigning ``$s_a \leftarrow q_\emptyset$''
  is correct.

  Otherwise, let $b \in \Sigma$ be the unique letter given by
  pre-compatibility for $a$. By the above equivalences, we have
  $(\Pre{K}{L})^a = \Pre{K^{(a, b)}}{L^b}$.

  Upon reaching Line~\ref{ln:pp:qq}, we have $\lang{p'} = K^{(a, b)}$
  and $\lang{q'} = L^b$. Therefore, if $p' = p$ and $q' = q$, then
  $(\Pre{K}{L})^a = \Pre{K^{(a, b)}}{L^b} = \Pre{K}{L}$, and hence
  assigning ``$s_a \leftarrow \self$'' is correct. Otherwise, since
  $(\Pre{K}{L})^a = \Pre{K^{(a, b)}}{L^b} =
  \Pre{\lang{p'}}{\lang{q'}}$, assigning ``$s_a \leftarrow \pre{$p',
    q'$}$'' is correct. \qed
\end{proof}

\begin{proposition}
  \Cref{alg:pre:comp} with memoization works in time $\O(m^2 n_p
  n_q)$, where $m \defeq |\Sigma|$ and $n_x$ is the number of nodes
  reachable from node $x$.
\end{proposition}

\begin{proof}
  Due to memoization, Line~\ref{ln:pre:comp:for} is reached at most
  once per pair $(x, y)$ where $x$ is a node reachable from $p$, and
  $y$ is a node reachable for $q$. Moreover, the loop iterates over
  each letter $a \in \Sigma$. The body of the loop needs to identify
  the unique letter $b$ or its nonexistence. This can be done naively
  by checking whether $p^{(a, b)} \neq \emptyset$ and $q^b \neq
  \emptyset$ for each $b \in \Sigma$. \qed
\end{proof}

\section{Missing proofs of \Cref{sec:model:check}}

\propLossy*

\begin{proof}
  Clearly, pure lossiness allows at least as many behavior as
  semi-lossiness, so the implication from right to left is
  obvious. Let us prove the other direction.

  Let $C \lcstrans{\pi}_\text{pure-lossy} C'' \succeq C'$. We proceed
  by induction on the number of $\lcsread{\cdot}{\cdot}$ occurring
  along $\pi$. If there is no occurrence, then we can remove all the
  losses to obtain a sequence $\pi'$. We clearly have
  \[
  C \lcstrans{\pi'}_\text{semi-lossy} D \text{ for some } D \succeq C''.
  \]

  Now, consider the general case. Let $\lcsloss{\sigma}{i}$ denote the
  loss of an occurrence of letter $\sigma$ on channel $i$ (without
  changing the state of any process). We have
  \[
  \pi = \pi_0\ \lcsloss{a_1}{j_1}\ \pi_1\ \lcsloss{a_2}{j_2}\ \cdots\ \lcsloss{a_m}{j_m}\ \pi_m\ \lcsread{\sigma}{i}\ \pi',
  \]
  where each $\pi_j$ is a sequence of $\lcswrite{\cdot}{\cdot}$ and
  $\lcsnop{}$ operations, and $\pi'$ is the rest of the sequence. The
  losses on channel $j_\ell \neq i$ do not impact
  $\lcsread{\sigma}{i}$, and so they can be postponed. Moreover, any
  $\lcsloss{a_\ell}{i}$ that affects a letter appearing \emph{after}
  the letter first occurrence of $\sigma$, can be postponed. Thus, we
  can reorganize $\pi$ as
  \[
  w \defeq \overbrace{\pi_0\ \pi_1\ \cdots\ \pi_m\
  \prod_{\mathclap{\substack{\ell \in [1..m] \\ j_\ell = i \\ a_j~\text{occurs before}~\sigma}}} \lcsloss{a_\ell}{i}\
  \lcsread{\sigma}{i}}^{u \defeq {\,}}\
  \overbrace{\prod_{\mathclap{\substack{\ell \in [1..m] \\ j_\ell = i \\ a_j~\text{occurs after}~\sigma}}} \lcsloss{a_\ell}{i}\
  \prod_{\mathclap{\substack{\ell \in [1..m] \\ j_\ell \neq i}}} \lcsloss{a_\ell}{j_\ell}\ \pi'}^{v \defeq {\,}}.
  \]
  Let $D$ be such that $C \lcstrans{u}_\text{pure-lossy} D
  \lcstrans{v}_\text{pure-lossy} C''$. We have $C
  \lcstrans{u}_\text{semi-lossy} D$. We are done by induction on
  $v$.

  \medskip

  It remains to prove the last part of the proposition, namely that
  ${\lcstrans{}_\text{semi-lossy}}$ is monotone. Let $C$ and $C'
  \succeq C$ be such that $C \lcstrans{t}_\text{semi-lossy} D$.

  If $t$ is a $\lcsnop$ or $\lcswrite{\cdot}{\cdot}$, then $C'
  \lcstrans{t}_\text{semi-lossy} D'$ where $D' \succeq C'$. Consider
  the case where $t = \lcsread{a}{i}$. Let $C = ([p_1, \ldots, p_m],
  w_1, \ldots, w_k)$ and $C' = ([p_1, \allowbreak \ldots, \allowbreak
    p_m], \allowbreak w_1', \ldots, w_k')$. Let $w_i = x a y$ where
  $x$ contains no $a$. Since $w_i \sqsubseteq w_i'$, we have $w_i' =
  x' a y'$ for some $x' \sqsupseteq x$ and $y' \sqsupseteq y$. If $x'$
  contains no $a$, then we are done since $\lcsread{a}{i}$ converts
  $w_i$ into $y$; $\lcsread{a}{i}$ converts $w_i'$ into $y'$; and $y
  \sqsubseteq y'$. Otherwise, $x' = u a v$ where $u$ contains no
  $a$. In this case, $\lcsread{a}{i}$ converts $w_i$ into $y$, and
  $w_i'$ into $avy'$. As $y \sqsubseteq y'$, we have $y \sqsubseteq
  avy'$, and hence we are done. \qed
\end{proof}

\section{Extra experimental results for \Cref{sec:exp}}

\begin{table}
  \begin{center}
  %\resizebox{\columnwidth}{!}{%
    \begin{tabular}{lrrrrrr}
\toprule
& \bml & \bmlCsre & \bmlMof & \bmlSi & \mcscm & \wadl \\
\midrule
\texttt{tcp\_simplest\_err} & 0.04 & 0.08 & 0.10 & 0.22 & 0.01 & 0.04 \\
\texttt{BAwCC} & 3.28 & 0.24 & 0.28 & 20.79 & \TO & 0.18 \\
\texttt{BAwCC\_enh} & 4.17 & 0.25 & 0.25 & 23.44 & 124.72 & 0.50 \\
\texttt{tcp\_simplest} & 0.03 & 0.09 & 0.11 & 0.07 & 0.34 & 0.06 \\
\texttt{BAwPC\_enh} & 1.22 & 0.08 & 0.08 & 6.75 & 0.10 & 0.19 \\
\texttt{peterson\_3} & \TO & 34.92 & 36.97 & 5.21 & \TO & 19.50 \\
\texttt{BAwPC} & 0.83 & 0.11 & 0.11 & 7.85 & \TO & 0.11 \\
\texttt{peterson\_4} & \TO & \SO & \SO & \TO & \TO & \TO \\
\texttt{simple\_server-bwerror} & 0.01 & 0.01 & 0.01 & 0.03 & 0.15 & 0.01 \\
\texttt{ring2} & 0.36 & 1.02 & 0.91 & 0.45 & \SO & 0.04 \\
\texttt{brp\_like\_modified} & 0.39 & 0.64 & 0.69 & 0.28 & 0.39 & 0.29 \\
\texttt{simple\_server} & 0.13 & 1.44 & 0.76 & 2.79 & 0.25 & 0.06 \\
\texttt{pop3} & 400.20 & \TO & 1.60 & 0.73 & 4.73 & 2.07 \\
\bottomrule
\end{tabular}%}
  \end{center}
  \caption{Time in seconds to solve each lossy channel system
    instance, where \TO~=~timeout and \SO~=~stack overflow.}
\end{table}

\begin{table}
\begin{center}
  \begin{minipage}[t]{0.46\textwidth}\vspace*{0pt}
%    \resizebox{\columnwidth}{!}{%
\begin{tabular}{>{\scriptsize}l>{\scriptsize}r}
\toprule
Broadcast protocol instance & \wadl \\
\midrule
\texttt{MOESI\_exclusivemodified} & 0.02 \\
\texttt{dragon\_exclusiveshared} & 0.12 \\
\texttt{firefly\_deadlock} & 0.01 \\
\texttt{MOESI\_sharedexclusive} & 0.01 \\
\texttt{dragon\_dirtyshareddirty} & 0.04 \\
\texttt{dragon\_dirtysharedexclusive} & 0.16 \\
\texttt{FutureBus\_exclusiveexclusive} & 65.32 \\
\texttt{Berkeley\_exclusiveexclusive} & \TO \\
\texttt{MOESI\_ownedmodified} & 0.02 \\
\texttt{dragon\_exclusivedirty} & \TO \\
\texttt{firefly\_dirtyexclusive} & \TO \\
\texttt{MOESI\_exclusiveexclusive} & 0.01 \\
\texttt{Berkeley\_deadlock} & 0.01 \\
\texttt{FutureBus\_deadlock} & 0.01 \\
\texttt{FutureBus\_pendingrightsecond} & 1.02 \\
\texttt{Berkeley\_exclusivenonexclusive} & 0.01 \\
\texttt{synapse\_dirtydirty} & 0.01 \\
\texttt{Berkeley\_exclusiveunowned} & 0.01 \\
\texttt{MOESI\_deadlock} & 0.01 \\
\texttt{dragon\_deadlock} & 0.01 \\
\texttt{dragon\_exclusiveexclusive} & 0.05 \\
\texttt{Illinois\_dirtyshared} & 0.06 \\
\texttt{MOESI\_ownedexclusive} & 0.01 \\
\texttt{firefly\_dirtyshared} & 0.05 \\
\texttt{firefly\_dirtydirty} & \TO \\
\texttt{FutureBus\_sharedexclusive} & 41.12 \\
\texttt{dragon\_dirtydirty} & \TO \\
\texttt{firefly\_exclusiveexclusive} & \TO \\
\texttt{MESI\_sharedmodified} & 0.01 \\
\texttt{synapse\_deadlock} & 0.01 \\
\texttt{MOESI\_modifiedmodified} & 0.01 \\
\texttt{FutureBus\_secondpending} & 0.79 \\
\texttt{Illinois\_dirtydirty} & 0.03 \\
\texttt{MESI\_modifiedmodified} & 0.01 \\
\texttt{MOESI\_sharedmodified} & 0.01 \\
\texttt{synapse\_dirtyvalid} & 0.01 \\
\texttt{Illinois\_deadlock} & 0.01 \\
\texttt{dragon\_shareddirty} & 0.24 \\
\bottomrule
  \end{tabular}%}
  \end{minipage}\hspace*{5pt}%
  \begin{minipage}[t]{0.54\textwidth}\vspace*{0pt}
%  \resizebox{\columnwidth}{!}{%
\begin{tabular}{>{\scriptsize}l>{\scriptsize}r}
\toprule
Regular model checking instance & \wadl \\
\midrule
\texttt{icalp-Israeli-Jafon\_notoken} & 0.01 \\
\texttt{Dijkstra-ring\_deadlock} & \TO \\
\texttt{token-passing\_notoken} & 0.01 \\
\texttt{icalp-Herman-ring\_notoken} & 0.01 \\
\texttt{DijkstraMutEx\_nomutex} & \TO \\
\texttt{Burns\_nomutex} & \UNK \\
\texttt{Herman\_notoken} & 0.01 \\
\texttt{Dijkstra-Scholten\_twob} & \UNK \\
\texttt{bakery\_deadlock} & 0.01 \\
\texttt{icalp-bakery\_nomutex} & 0.01 \\
\texttt{Burns\_deadlock} & 0.01 \\
\texttt{mux-array\_deadlock} & \TO \\
\texttt{icalp-Herman-line\_notoken} & 0.01 \\
\texttt{dining-cryptographers\_deadlock} & 0.01 \\
\texttt{token-passing-no-invariant\_deadlock} & 0.01 \\
\texttt{token-passing-no-invariant\_manytoken} & \UNK \\
\texttt{icalp-Israeli-Jafon\_deadlock} & 0.01 \\
\texttt{token-passing\_manytoken} & 0.01 \\
\texttt{Szymanski\_nomutex} & \TO \\
\texttt{res-allocator\_deadlock} & \TO \\
\texttt{icalp-bakery\_deadlock} & 0.01 \\
\texttt{DijkstraMutEx\_deadlock} & \TO \\
\texttt{token-passing-no-invariant\_nonreach} & \UNK \\
\texttt{token-passing-no-invariant\_notoken} & 0.01 \\
\texttt{bakery\_nomutex} & 0.01 \\
\texttt{icalp-Herman-ring\_deadlock} & 0.01 \\
\texttt{icalp-Herman-line\_deadlock} & 0.01 \\
\texttt{token-passing\_deadlock} & 0.01 \\
\texttt{mux-array\_nomutex} & \UNK \\
\texttt{Israeli-Jafon\_notoken} & 0.01 \\
\texttt{Dijkstra-Scholten\_twod} & \UNK \\
\texttt{Herman\_deadlock} & 0.01 \\
\texttt{Israeli-Jafon\_deadlock} & 0.01 \\
\texttt{Dijkstra-ring\_nomutex} & \UNK \\
\texttt{res-allocator\_nomutex} & 0.01 \\
\bottomrule
\end{tabular}%}
\end{minipage}  
\end{center}
  \caption{Time in seconds to solve each broadcast protocol instance
    (left) and regular model checking instance (right), where
    \TO~=~timeout and \UNK~=~unknown. The ``unknown'' instances are
    those where a cycle was contracted in Line~\ref{ln:S:marked} of
    \Cref{alg:pre:gen}, and hence where we cannot be certain that the
    result is weakly acyclic.}
\end{table}

\begin{figure}
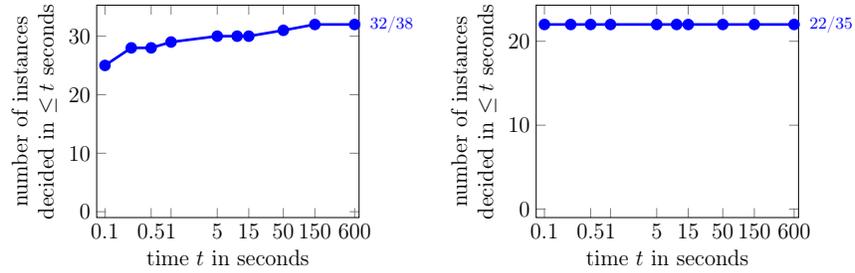

  \centering
  \begin{minipage}[t]{0.47\textwidth}
    \begin{center}
      \plotcanvas{time $t$ in seconds}{number of instances \\ decided in $\leq t$ seconds}{%
        \toolplotwithtotalcount{wadl}{\markwadl}{
          (0.1, 25)(0.25, 28)(0.5, 28)(1.0, 29)(5.0, 30)(10.0, 30)(15.0, 30)(50.0, 31)(150.0, 32)(600.0, 32)
        }{600.0}{32}{32/38};
      }
    \end{center}
  \end{minipage}
  \begin{minipage}[t]{0.47\textwidth}
    \begin{center}
      \plotcanvas{time $t$ in seconds}{number of instances \\ decided in $\leq t$ seconds}{%
        \toolplotwithtotalcount{wadl}{\markwadl}{
          (0.1, 22)(0.25, 22)(0.5, 22)(1.0, 22)(5.0, 22)(10.0, 22)(15.0, 22)(50.0, 22)(150.0, 22)(600.0, 22)
        }{600.0}{22}{22/35};        
      }
    \end{center}
  \end{minipage}

  \caption{Cumulative number of instances decided by \wadl{} over time
    (semi-log scale) for broadcast protocols (left) and for the other
    regular model checking instances (right).}
  \label{fig:results:dodo}
\end{figure}

\end{document}